%% file: main_supermaps.tex
\documentclass[a4paper,onecolumn,11pt,unpublished]{quantumarticle}
\pdfoutput=1

\input{preamble.tex}
\usepackage{mathtools}
\usepackage[utf8]{inputenc}
\usepackage[english]{babel}
\usepackage[T1]{fontenc}
\usepackage{comment}

\usepackage{amsthm}
\usepackage{thmtools}
\usepackage{thm-restate}
\newtheorem{theorem}{Theorem}
\newtheorem{lemma}{Lemma}
\newtheorem{corollary}{Corollary}

\theoremstyle{definition}
\newtheorem{definition}{Definition}
\newtheorem{example}{Example}

\theoremstyle{remark}
\newtheorem{remark}{Remark}

\pgfsetlayers{background,strings,edgelayer,nodelayer,foreground,main}

\begin{document}

\title{Supermaps on generalised theories}

\author[1]{Matt Wilson}
\author[2]{James Hefford}
\author[3]{Timothée Hoffreumon}
\affiliation{Université Paris-Saclay, CNRS, ENS Paris-Saclay, Inria, CentraleSupélec,  Laboratoire Méthodes Formelles, France.}
\affiliation[3]{Mathematical Institute, Slovak Academy of Sciences, \v{S}tef\'{a}nikova 49, 814 73 Bratislava, Slovakia.}

\maketitle

\begin{abstract}
Categorical supermaps generalise higher-order quantum operations from finite-dimensional quantum theory to arbitrary circuit theories. In this paper, we establish the Yoneda lemma for categorical supermaps, which states that whenever a physical theory has a suitable notion of channel-state duality, then categorical supermaps on that theory can be concretely represented in terms of that duality. This lemma eliminates any guesswork or ambiguity when defining the appropriate notion of supermap for these theories. As a concrete application, we show that the categorical supermaps on Boxworld are in general characterised by a physically well-motivated weaking of the recently proposed NSWSE principle for higher-order processes on boxworld.
Furthermore, via the same Yoneda lemma we put forward a stable definition for supermaps in real quantum theory.
\end{abstract}

In standard quantum information theory, the information resources are typically represented by quantum states, while the transformations they undergo are represented by quantum channels (see e.g. \cite{Wilde_2013}). In this framework, the spatial and temporal compositions arise naturally from, respectively, the tensor product and channel composition, forming the circuit model of quantum information \cite{deutsch1989quantum}. However, this `bottom-up' approach to spacetime compatibility overlooks the broader, `top-down' perspective: rather than `building' the environment around resources by composing channels, one can instead consider every possible spacetime environment compatible with the quantum mechanics locally obeyed by the resources \cite{chiribella_switch, oreshkov}. This has motivated a shift from studying quantum channels to studying higher-order quantum operations \cite{taranto2025higherorderquantumoperations}.

Such higher-order operations are often referred to as \textit{quantum supermaps} \cite{chiribella_supermaps, chiribella_circuits, chiribella_networks} or \textit{process matrices} \cite{oreshkov}. These treat channels as the fundamental resources to be transformed, capitalising on state preparation being a special kind of channel. 
This concept is in fact quite canonical, as evidenced by its independent rediscovery through various motivations \cite{Gutoski_2007,Ziman_2008,process_tensor, chiribella_supermaps, oreshkov}. It is therefore natural to ask whether this notion can be generalised from quantum theory to arbitrary physical theories cast in a suitably close framework, like \textit{Generalised Probabilistic Theories} (GPTs) \cite{barrett_gpts}, \textit{Operational Probabilitistic Theories} (OPTs) \cite{chiribella_informational}, or \textit{Categorical Probabilistic Theories} (CPTs) \cite{gogioso_cpt} for instance.

There are at least two compelling reasons to pursue such a generalisation.
First, just as generalised probabilistic theories have probed the fundamental limits of physically achievable correlations under a no-signalling constraint \cite{Popescu:1994aa}, a stable higher-order framework could be used to study the fundamental limits of causation, for example by studying the higher-order theory of logically consistent no-signalling constraints \cite{bavaresco2024indefinitecausalorderboxworld}.
Second, even within quantum theory itself, the appropriate notion of higher-order processes for infinite-dimensional systems remains unclear. Several alternative definitions exist, including \cite{Chiribella2013NormalCP, Giacomini_2016, wilson_locality, wilson_polycategories}, but their interrelationships are unexplored, hindering connections with quantum field theory \cite{paunkovic2023challenges} as well as the traditional approaches to quantum gravity.

Several works have begun to explore the generalisation of higher-order processes beyond quantum theory. 
Early works on indefinite causal order beyond quantum theory discovered that even the generalisation to classical theory leads to non-causal yet logically consistent environments \cite{baumeler_incompatibility, baumeler_logically}.
The definitions of higher-order processes on both classical and quantum theory were then generalised using the concept of a \textit{pre-causal category}, leading to a first categorical definition of higher-order processes on generalised theories \cite{kissinger_caus}.
This approach introduced pre-causal categories in order to abstract some of the core features of classical and quantum theory and then provided a general recipe for obtaining a higher-order theory from a pre-causal category by applying a canonical construction from linear logic. This construction establishes that higher-order processes are models of linear and spatiotemporal logics \cite{SimmonsKissinger2022, simmons_completelogic} --- in other words, that the different ways to compose processes together reflects the signalling relations between the composed processes \cite{apadula2022nosignalling,hoffreumon2022projective,jencova2024structurehigherorderquantum} --- thereby taming their composition rules \cite{SimmonsKissinger2022, simmons_completelogic, apadula2022nosignalling, Bisio_2019, jencova2024structurehigherorderquantum, hoffreumon2022projective}. 

However, this construction is only as general as the abstract structure it was built upon: while recent works showed that the notion of pre-causality is not necessary to define supermaps \cite{wilson_locality, wilson_polycategories, jencova2024structurehigherorderquantum}, they still rely on the weaker notion of \textit{compact closedness}. In a nutshell, requiring supermaps to be defined over compact closed categories is a technical abstraction of saying that supermaps can only be defined on theories admitting a notion of channel-state duality. In the quantum case, it means that the definition of supermaps relies on having a well-defined Choi-Jamio{\l}kowski (CJ) isomorphism \cite{Jamiolkowski_1972,Choi_1975}. This assumption remains quite limiting, excluding, for instance, infinite-dimensional quantum theory.

`Categorical Supermaps' is the phrase we will use to refer to a series of works \cite{wilson_locality, wilson_polycategories, hefford_supermaps, hefford2025bvcategoryspacetimeinterventions} which aim to evade these pitfalls by relying only on a single assumption: that the physical theory forms a \textit{symmetric monoidal category} \cite{coecke_picturalism}, meaning that the physical theory comes equipped with a reasonable notion of `parallel composition of systems' --- think of this as the tensor product of Hilbert spaces in the case of quantum theory --- allowing one to define a circuit theory. The core principle of categorical supermaps is to formalise higher-order maps as families of functions called \textit{locally-applicable transformations}, where `local-applicability' encodes certain commutativity requirements between the transformations \cite{wilson_locality, wilson2023originlinearityunitarityquantum}. A central theorem of these works shows that these families are in bijection with quantum and classical supermaps when respectively applied to the information theory of finite-dimensional quantum and classical systems. Moreover, by using the theory of profunctors \cite{loregian_coend} and the abstraction of iterative transformations (colloquially called `combs') in terms of profunctor optics \cite{roman_coend, roman_comb, boisseau_cornering, earnshaw, hefford_optics}, the categorical supermaps framework generalises the logic of higher-order quantum theory \cite{SimmonsKissinger2022, simmons_completelogic, apadula2022nosignalling, Bisio_2019, jencova2024structurehigherorderquantum, hoffreumon2022projective} to any circuit theory \cite{hefford_supermaps,hefford2025bvcategoryspacetimeinterventions}: the rules to compose the higher-order operations on any circuit theory follow a model of logic close to BV-logic \cite{guglielmi}, supporting currying, negation, the tensor and par of linear logic, as well as a non-commutative sequential connective modelling definite causal order.

Recent work on the generalisation of higher-order processes beyond quantum theory has focused on constructing higher-order maps over specific generalised probabilistic theories, such as Boxworld \cite{bavaresco2024indefinitecausalorderboxworld}, Hex-Square Theory \cite{sengupta2024achievingmaximalcausalindefiniteness}, and time-symmetric quantum theory \cite{mrini2024indefinitecausalstructurecausal,Chiribella_Liu}, as well as identifying information-theoretic principles expected to be common amongst higher-order processes on broader classes of theories \cite{liu2025parityerasurefoundationalprinciple}. The Boxworld supermap definition of \cite{bavaresco2024indefinitecausalorderboxworld} is particularly subtle: in this context, Boxworld is the circuit theory whose states are classical channels and whose parallel composition is the no-signalling composition\footnote{We follow the definition of \cite{bavaresco2024indefinitecausalorderboxworld} but remark that it is not the most general formulation of Boxworld that can be found in the literature \cite{barrett_gpts}.}, so Boxworld operations are already a special kind of higher-order operation between classical systems. The authors argue that a physically meaningful notion of supermap on Boxworld operations, in particular the one corresponding to a process matrix, which we will call \textit{Boxworld tensor}\footnote{The authors of this work actually use the nomenclature `process tensor'; we prefer `Boxworld tensor' so as to avoid using the word `process' in too many contexts.}, requires a constraint called `No-Signalling Without System Exchange (NSWSE)'. Since Boxworld is a monoidal category, and one with a CJ-isomorphism, one can apply the categorical supermaps framework \cite{wilson_locality, wilson_polycategories, hefford_supermaps, hefford2025bvcategoryspacetimeinterventions}. This leaves open a natural concrete question: do categorical supermaps for Boxworld recover NSWSE?

We use this question to illustrate the value of formal logical and categorical methods for studying indefinite causal order in general physical theories, and to motivate a general structural characterisation theorem for categorical supermaps on theories with CJ-isomorphisms. First, we show that the CJ-based definition of supermaps on Boxworld returns a physically well-motivated weakening of the NSWSE principle via mostly type-theoretic reasoning using the BV-structure of higher-order classical theory \cite{SimmonsKissinger2022, simmons_completelogic}. Then, we consider a broader, more structural question: under what conditions can categorical supermaps be characterised via a CJ-isomorphism? Addressing it leads to n adaptation of the Yoneda Lemma to categorical supermaps, which in turn demonstrates that categorical supermaps on Boxworld satisfy the same relaxation of the NSWSE principle.

The Yoneda Lemma of category theory is indeed the result offering such a concrete representation \cite{maclane1998categories}. In technical language, this fundamental lemma states that natural transformations between representable presheaves on a category correspond to morphisms in the category. Concretely, it means that when a family of functions on processes commutes with pre-composition, it can be represented by post-composition with a process. At first glance, categorical supermaps do not seem to support such a representation theorem: defining supermaps on the category of unitary channels provides a counterexample \cite{wilson_polycategories}. However, we find that while there cannot be a Yoneda Lemma for categorical supermaps on arbitrary monoidal categories, there is one for categorical supermaps on physical theories with a Choi isomorphism. More precisely, by modelling general theories as a relaxation of categorical probabilistic theories \cite{gogioso_cpt}, we establish the following main theorem:

\begin{restatable}[Yoneda Lemma for Categorical Supermaps]{theorem}{YonedaLemma}\label{thm:yoneda}
Let $\mathcal{C}$ be a generalised theory with a channel-state duality, then:  \[ \text{Categorical supermaps} = \text{CJ supermaps} .\]
\end{restatable}
As a consequence, we prove that the categorical supermap formalism recovers supermaps on classical and quantum theory, while putting forward a physically well-motivated extension of the NSWSE definition for higher-order maps on Boworld:

\begin{restatable}[Recovery of Known Supermap Approaches]{theorem}{Recovery}
The following hold:
\begin{itemize}
\item Categorical Supermaps on Classical Theory = Classical Supermaps.
\item Categorical Supermaps on Quantum Theory = Quantum Supermaps.
\item Categorical Supermaps on Boxworld $\cong $ cw-NSWSE Boxworld Tensors.
\end{itemize}
\end{restatable}

More generally, as soon as a theory supports a concrete supermap definition, our abstract categorical approach is equivalent to it. Since categorical supermaps represent the minimal definition of supermaps, this concrete representation theorem demonstrates a surprising stability in defining supermaps for general physical theories. For example, it shows that a well-motivated extension of the NSWSE requirement naturally arises from applying the categorical supermaps approach to Boxworld. Moreover, it equips the study of higher-order processes in Boxworld with tools for iterated higher-order objects via currying and composition rules based on the established BV structure of categorical supermaps \cite{hefford2025bvcategoryspacetimeinterventions}.

As a novel application of Theorem \ref{thm:yoneda}, we note that real quantum theory is a generalised theory with channel-state duality. 
This allows us to present the first treatment of supermaps in real quantum theory while providing further evidence of the canonicity of categorical supermaps.

\begin{corollary}
    When $\cat{C}$ is real quantum theory (or quantum theory defined over any field extension of the constructible numbers), the categorical supermaps coincide with the CJ supermaps.
\end{corollary}

Therefore, by establishing a version of the Yoneda Lemma suited for categorical supermaps, we clarify the landscape of possible causal structures in general physical theories while leveraging a core tool of category theory. To enhance readability, categorical notions are explained using diagrammatic notation and equations throughout. It is the authors' intention to make this paper fairly accessible to readers who are not experts in category theory so to get them accustomed to the categorical (supermaps) framework.

\section{Generalised Theories}


There are many frameworks for general physical theories, including the aforementioned Generalised Probabilistic Theories (GPTs) \cite{barrett_gpts}, Operational Probabilistic Theories (OPTs) \cite{chiribella_purification} and Categorical Probabilistic Theories (CPTs) \cite{gogioso_cpt}.
In this work, we will define a notion of generalised theory which sits between CPTs and OPTs. Its core postulate is that the classical variables representing the outcomes of measurements as well as the settings conditioning the agents' actions can be represented within the theory, and in a circuit-theoretic way using boxes and wires \cite{gogioso_cpt}. 
Every CPT has this property, in addition to more stringent causality properties, so the results of this paper apply directly to all CPTs. 
On the other hand, it is an open question as to whether every OPT can be re-interpreted as a generalised theory in the sense in which they are presented here.
In fact, the closest existing notion in the literature to what we will refer to as generalised theories is the one of `theories with classical interfaces' as defined in \cite{Selby2021reconstructing}.

Before presenting the technical definition, we pause to provide some intuition of the categorical concepts that will be used. The following paragraphs are aimed at the reader knowledgeable in quantum theory but not familiar with its category-theoretical rephrasing.

A category $\cat{C}$ is a collection of `objects', often denoted via capital Latin letters $A,B,\dots$, together with a collection of `morphisms', often denoted with lower-case Latin letters $f,g,\dots$, which describe how the objects are interrelated with one another. 
In quantum theory, objects are typically taken to be the Hilbert spaces associated with every system. Everything else is encoded by morphisms, which are typically taken to be some specified subset of the linear maps between Hilbert spaces.
A commonly used notion of morphism is that of a (quantum) channel from an input system $A$ to an output system $B$. Such a channel can be thought of as a morphism $f: A \rightarrow B$ between objects $A$ and $B$ in the category of quantum channels between Hilbert spaces. Accordingly, we say that this morphism is `of type $A \rightarrow B$'. Morphisms naturally come with a notion of sequential composition: out of any two morphisms, say $f$ of type $ A \rightarrow B$ and $g$ of type $B \rightarrow C$, there exists a morphism $g \circ f$ of type $A \rightarrow C$. For quantum channels, this sequential composition operation is simply function composition. 


It is important to understand that, in the category of quantum channels, state preparation and post-selection (components of destructive measurements) are special kinds of morphisms, not objects. To represent them, we use \textit{monoidal categories}. Every monoidal category $\cat{C}$ contains a `trivial object', for which we reserve the symbol $I$.
In the category of quantum channels, $I$ is the one-dimensional system and the preparation of a system $A$ in a given state $\rho$ is a morphism $f_\rho$ from the trivial object $I$ (representing ``nothing'') to the object $A$ (representing the Hilbert space where the state is defined); $f_\rho$, then, has type $I \rightarrow A$. Generalising, a state is a morphism from the trivial object to any other object; picking a state $\rho$ in Hilbert space $A$ becomes picking a morphism $f$ of type $I \rightarrow A$. Similarly, a (destructive) measurement effect is a morphism from any object to the trivial object. For example, the morphism $g_{e_i}$ of type $A \rightarrow I$ represents the application of a measurement effect (i.e. POVM element) $e_i$ to whichever state was in $A$. The composition of morphisms is what gives the Born rule: the application of the POVM element $e_i$ to state $\rho$ has become the composition of morphisms $g_{e_i} \circ f_\rho$. From their types $(A \rightarrow I) \circ (I \rightarrow A)$, we can deduce that the result has type $I \rightarrow I$. The probability of measuring the effect $e_i$ from state $\rho$ is therefore itself a morphism from the trivial system to itself, and such morphisms are typically referred to as scalars. The usual Born rule $p(i) = \Tr[\rho \: e_i]$ is then substituted by the following composition of morphisms $h_{p(i)} = g_{e_i} \circ f_\rho$.


Hence, state preparation, evolution, and measurement are all sorts of morphisms contained in the category, while the kind of systems these procedures (which we will more often call `transformations', `operations', or `processes') relate with one another are the objects in the category.

Monoidal categories not only model the sequential composition of processes, but also the parallel composition of systems and processes alike.
The tensor product of Hilbert spaces returns a Hilbert space, and the tensor product of two linear maps between Hilbert spaces returns a linear map between Hilbert spaces. Monoidal categories model the salient features of this tensor product by supporting a parallel composition operation $\otimes$. In the monoidal category of linear maps between Hilbert spaces, the `tensor product between the Hilbert spaces associated with systems A and B' corresponds to a new object $A \otimes B$, while the notion of `Kronecker product between the linear maps $f$ and $g$' accordingly corresponds to a new morphism $f \otimes g$. Briefly put, when a category possesses such an extra operation $\otimes$ with unit $I$, it is called a monoidal category. And when a monoidal category also contains swap morphisms $A \otimes B \rightarrow B \otimes A$, it is called a \textit{symmetric monoidal category}.


The category of completely positive maps is contained as a subcategory of the (symmetric monoidal) category of Hilbert spaces: to obtain it, one needs to restrict the objects (i.e., Hilbert spaces) to only the Hilbert spaces of (Hilbert-Schmidt) operators, and restrict the set of all morphisms (i.e., linear maps between Hilbert spaces) to only the completely positive maps. 
Adding determinism and a time direction, the category of quantum channels is a further restriction requiring the morphisms to be trace-preserving, so as to limit them to the usual notions of quantum states, effects, and channels. 

One can subsequently study similar theories by changing the (symmetric monoidal) category: for instance, the information theory of classical systems (which we call `classical theory') is analogously defined as the category of stochastic maps between real vector spaces (or, alternatively, as the category of stochastic matrices). Thinking about these quantum(-like) theories in terms of types is then very well-suited to define higher-order maps on these: if a quantum channel is a morphism of type $A \rightarrow B$, then the quantum supermap must have type $(A \rightarrow B) \rightarrow (A' \rightarrow B')$. However, this expression is outside of the category since $A \rightarrow B$ and $A' \rightarrow B'$ are a priori not objects themselves. To build higher-order theories, therefore, one must show how to build categories in which these would make sense as objects. Making this naive type-reasoning rigorous is exactly the point of categorical supermaps and the related adaptation of the Yoneda lemma developed in the present paper. 

In the following, we will use the term \textit{processes} to refer to morphisms between any two non-trivial objects, \textit{states} for the morphisms of type $I \rightarrow A, \forall A \neq I$, and \textit{effects} for those of type $A \rightarrow I, \forall A \neq I$. 
With this in mind, the reader should now have an overview of the material needed to grasp our definition: a \textit{generalised theory} is a symmetric monoidal category $\cat{C}$ (in fact, as we will see, a collection of symmetric monoidal categories) whose objects can be think of as the physical systems, and whose morphisms as the physical processes. The sequential composition of morphisms $g\circ f$ is its timelike composition, and the monoidal composition $g\otimes f$ is its spacelike one. 
But generalised theories are not only symmetric monoidal categories; their core postulate has yet to be implemented. It states that a generalised theory must also contain special objects, which we will denote with blackboard font $\mathbb{X}$, such that the processes of the form $\mathbb{X}\morph{}\mathbb{Y}$ are non-normalised stochastic matrices. These special objects can be thought of as the ``decohered systems'', i.e., the systems that will behave classically. 
We now give a rather formal definition to fix terminology for those more familiar with category theory. A purely diagrammatic description of generalised theories is given directly after, and is all that is needed to follow the remaining proofs and definitions of the paper.

\begin{definition}[Generalised Theories]
    A generalised theory is a convex enriched symmetric monoidal category $\cat{C}$, representing the non-deterministic physical events, and equipped with:
    \begin{itemize}
        \item a specified wide convex subcategory $\cat{C}^{\mathbf{d}} \hookrightarrow  \cat{C}$, representing the deterministic physical events,
        \item a full and faithful embedding $\stoch \hookrightarrow \cat{C}^{\mathbf{d}}$,
        \item a full and faithful embedding $\mat{\reals^+} \hookrightarrow \cat{C}$,
    \end{itemize}
    such that the following diagram commutes,
    \xsavebox{cground}{$\cat{C}^{\mathbf{d}}$}
    \[
    \begin{tikzcd}
        \stoch \arrow[r] \arrow[d] & \xusebox{cground} \arrow[d] \\
        \mat{\reals^+} \arrow[r] & \cat{C}
    \end{tikzcd}
    \]
    and such that all processes can be classically controlled. That is, for any family $\{M_{i}: A\rightarrow A'\}_i$ of deterministic processes (meaning the morphisms of $\cat{C}^{\mathbf{d}}$), there exists a deterministic physical process $M: A\otimes \mathbb{I}  \rightarrow A'$ and a deterministic classical process $i: I \rightarrow \mathbb{I}$ such that
    \begin{equation}\label{eq:control}
    	M \circ (1_A\otimes i)  = M_i \:.
    \end{equation}
    \end{definition}

The heuristic picture is as follows: a generalised theory consists of a deterministic theory $\cat{C}^{\mathbf{d}} $, whose morphisms represent physically admissible processes that can be achieved with certainty (like e.g., any state preparation or transformation). The deterministic theory is included in a non-deterministic theory $\cat{C}$, $\cat{C}^{\mathbf{d}}  \subseteq \cat{C}$, whose morphisms represent all physically admissible operations, including those that can only be achieved probabilistically (like e.g., a measurement outcome or post-selection).  Furthermore, deterministic classical theory (represented as `$\stoch$', the category of stochastic matrices) is included in the deterministic physical theory $\cat{C}^{\mathbf{d}} $ and non-deterministic classical theory (represented as `$\mat{\reals^+}$', the category of positive-valued matrices) is included in the non-deterministic theory.

To represent this using circuit diagrams, the usual conventions will be in use: diagrams are read from bottom to top, non-trivial systems (objects) are represented with wires, processes (morphisms) with boxes, and the classical systems will be coloured in blue to distinguish them from the physical systems. In the following picture, we represent the inclusion $\cat{C}^{\mathbf{d}}  \subseteq \cat{C}$, in terms of morphisms.
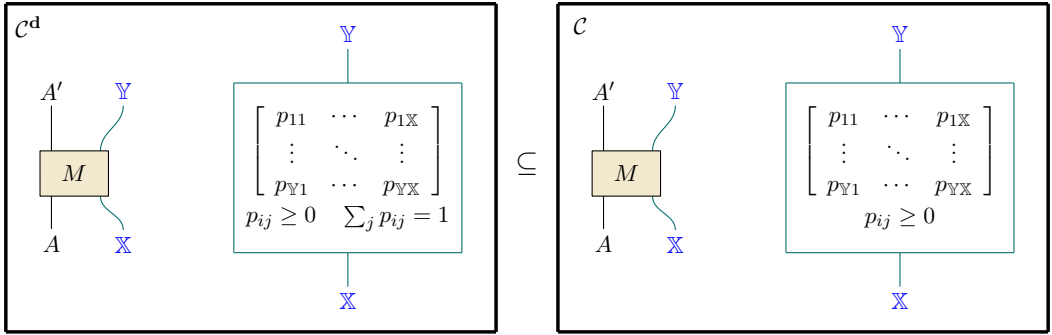
\[ \input{figs/cpt_cartoon_2.tikz} \ \subseteq \ \input{figs/cpt_cartoon_1.tikz} \]
In the above, both categories are shown to contain a physical process $M$ of type $(A \otimes \mathbb{X}) \rightarrow (B \otimes \mathbb{Y})$ (the boxes on the left side of each frame) together with a classical process of type $\mathbb{X} \rightarrow \mathbb{Y}$ (right side of each frame). Each of these classical processes has been ``zoomed in'' to make apparent the respective stochastic and positive real matrices they represent.
 
In summary, a generalised theory consists of four interlocked monoidal categories with the following interpretations.
\begin{itemize}
    \item $\mat{\reals^+}$: The symmetric monoidal category of positive real-valued matrices, whose objects are the natural numbers $n\in\nats$ and morphisms $n\morph{}m$ are the $m\times n$ matrices with entries from $\reals^+$. This is the monoidal category in which basic (not necessarily normalised) probability calculations can be made. 
    \item $\stoch$: The symmetric monoidal category of stochastic matrices (finite-dimensional classical channels). The objects are again the natural numbers $n\in\nats$ and a morphism $n\morph{}m$ is an $m\times n$ stochastic matrix. This is the monoidal category which tracks the causal flow of classical information.
    \item $\cat{C}$: A symmetric monoidal category representing every event (possibly the non-deterministic ones as well) of the physical theory, i.e.\ a mathematical description of post-measurement results\footnote{In principle, this can contain more than just measurement outcomes as it can contain super-normalised processes.}. The processes of this category are closed under convex combinations\footnote{Formally, this means that the homsets are convex subsets of vector spaces}.
    \item $\cat{C}^{\mathbf{d}}$: A symmetric monoidal category of deterministic events, those which can occur with probability $1$. The deterministic events are also closed under convex combinations. 
\end{itemize}
The inclusion of classical theory in whichever generalised theory is under consideration guarantees that some of the systems will obey classical information theory. This is a prerequisite for the last requirement: the possibility to classically control any event. 
This final requirement ensures that the classical theory blends properly within $\cat{C}$; the interaction \eqref{eq:control} between a deterministic classical state $x$ and family of deterministic processes $\{M_x\}_x$ it controls is diagrammatically represented as follows:
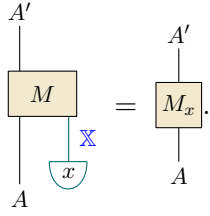
\[  \input{figs/cpt_control_1.tikz} \ = \ \input{figs/cpt_control_2.tikz}.  \]



Let us see some examples.
\begin{example}[Classical theory]
In classical theory we simply take $\cat{C}^{\mathbf{d}}  = \stoch$ and $\cat{C} = \mat{\reals^+}$. A deterministic process $f:n_A \rightarrow n_B$ is an $n_B \times n_A$ stochastic matrix, a deterministic state $\rho : I \rightarrow A$ is an $n_A$-dimensional probability vector, and a deterministic effect $E: A \rightarrow I$ is the transpose of the $n_A$ vector containing $1$'s in every entry.
\end{example}
\begin{example}[Quantum theory]
In quantum theory, we take $\cat{C}^{\mathbf{d}} $ to be the category of hybrid classical-quantum channels, that is, a morphism of type $(A   , n ) \rightarrow (B , m) $ with $n, m$ natural numbers is taken to be a function from elements of $[1, \dots ,n]$ to $m$-component instruments. More precisely, a morphism of type $(A,  n) \rightarrow (B, m) $ is taken to be a family $\{M_x^{y}\}_{x,y}$ of completely positive maps $M_x^y:A\rightarrow B$ such that for each fixed $x$ (often called a \textit{setting}) the family $\{M_x^{y}\}_{y}$ is a quantum instrument, meaning that $\sum_y M^y_x$ is a CPTP map. 
We take $\cat{C}$ to be the category of hybrid quantum-classical maps, meaning families $\{M_x^{y}\}_{x,y}$ of completely positive maps without any trace-preservation requirement. 
Sequential composition of $M: (A ,   n) \rightarrow (B,  m)$ with $N: (B,  m) \rightarrow (C ,  q)$ is given by the family $(N \circ M)_x^z = \sum_y N_y^{z} \circ M_x^{y}$. 
The parallel composition of systems is given by $(A ,  n) \otimes (B ,  m) := (A \otimes B ,n m)$.
The parallel composition of hybrid classical-quantum maps is given by $(N \otimes M)_{x,x'}^{y,y'} = \sum_y N_x^{y} \otimes M_{x'}^{y'}$. The monoidal unit is given by $(\mathbb{C} ,1)$.
Finally, the classical systems are given by those of the form $(\mathbb{I} ,n)$, since morphisms of type $(\mathbb{I} ,n) \rightarrow (\mathbb{I} ,m)$ are then matrices of scalars $M_x^y :  \mathbb{I} \rightarrow \mathbb{I}$, thus matrices valued (up to isomorphism) in the non-negative reals. In the deterministic case, these matrices come with the extra requirement that $\sum_y M_x^y = 1$, which indeed is the requirement that the matrix be stochastic. 
\end{example}
\begin{example}[Quantum theory on general subsystems]
In quantum theory with general subsystems, we take $\cat{C}^{\mathbf{d}} $ to be the category of quantum channels between finite-dimensional $C^*$-algebras\footnote{Recall that a (complex) $C^*$-algebra is but an abstraction of the algebra of linear operators over a (complex) Hilbert space.}. We take $\cat{C}$ to be the category of completely positive maps between $C^*$-algebras. The embedding $\stoch \hookrightarrow \cat{C}$ is given by identifying an $n$-dimensional basis of the variables of type $I \rightarrow \mathbb{X}$ (probability vectors) with a choice of an $n^2$-dimensional basis of the corresponding variables of type $I \rightarrow A$ (density matrices), such that the classical basis correspond to the diagonal entries of the quantum basis. Accordingly, the stochastic matrices of type $\mathbb{X} \rightarrow \mathbb{Y}$ are identified with the classical channels of type $A \rightarrow B$. That way, $\mathbb{X}$ is identified with the $C^*$-algebras of operators diagonal in a given basis of a Hilbert space. The projection from a morphism of $A$ to one of $\mathbb{X}$ thus recovers the notion of decoherence in a given basis. A similar procedure is conducted for the embedding $\mat{\reals^+} \hookrightarrow \cat{C}$.
\end{example}

\begin{example}[Boxworld (the theory of non-signalling classical boxes)]
Following the Boxworld definition of \cite{bavaresco2024indefinitecausalorderboxworld}, we define it as the generalised theory whose states are the morphisms of classical theory, but under the restriction that the monoidal structure (the tensor product) is no-signalling. In this somewhat nonstandard view, boxworld is defined as a subtheory of second-order classical theory. Concretely, it means that the states of Boxworld correspond to morphisms $(a \rightarrow a')$ of classical theory, with $a,a' $ classical objects and with the monoidal unit given by $I_{\texttt{Box}} = (1 \rightarrow 1)$. A state of boxworld then has the type $I_{\texttt{Box}} \rightarrow (a \rightarrow a')$. More precisely, the deterministic states of boxworld with type $I_{\texttt{Box}} \rightarrow (a \rightarrow a')$ correspond (up to a choice of representation we will soon make explicit) to  the stochastic matrices, and the non-deterministic states correspond to the non-negative matrices. 
But, and this is where Boxworld differs from the theory of second-order classical channels, the $n$-partite states of Boxworld $I_{\texttt{Box}} \rightarrow ((a \rightarrow a') \otimes (b \rightarrow b') \otimes \ldots)$ do not correspond to the full set of $ n$-partite morphisms of classical theory, which has type $(a \otimes b \otimes \ldots ) \rightarrow (a' \otimes b' \otimes \ldots)$. Rather, these are taken as the subset defined by the no-signalling condition (see \cite{bavaresco2024indefinitecausalorderboxworld} for the exact implementation of the condition),  which can be proven to have type $ ((a \rightarrow a') \otimes (b \rightarrow b') \otimes \ldots) \subset (a \otimes b \otimes \ldots ) \rightarrow (a' \otimes b' \otimes \ldots)$. 
Accordingly, the morphisms of Boxworld should correspond to linear mappings between no-signalling classical channels that (completely) preserve the no-signalling condition. 

We now state precisely the representation of Boxworld as outlined above that we will use in this paper. The objects of Boxworld are taken to be the symbols $\otimes_i (a_i \rightarrow a_i')$ with $a_i, a_i'$ positive integers, the morphisms of type $\otimes_i (a_i \rightarrow a_i') \rightarrow  \otimes_i (b_i \rightarrow b_i')$ are taken to be the real non-negative tensors (equivalently matrices, but referred to as tensors to disambiguate from our interpretation of states and matrices) $T_{b_1 \dots b_m a_1' \dots a_n'  }^{b_1'\dots b_m' a_1 \dots a_n}$. The deterministic morphisms are taken to be those tensors which furthermore preserve the no-signalling channels, meaning that, the contraction $T_{b_1 \dots b_m a_1' \dots a_n'  }^{b_1'\dots b_m' a_1 \dots a_n} M_{a_1 \dots a_n}^{a_1' \dots a_n'}$ with any non-signalling channel $M_{a_1 \dots a_n}^{a_1' \dots a_n'}$ returns a new non-signalling channel (i.e., an element of $\otimes_i (b_i \rightarrow b_i')$). It is easy to see from this definition that the states of this theory are indeed isomorphic to the classical non-signalling channels. The embedding of classical theory into Boxworld then works as follows, each classical system $a$ is identified with the Boxworld system $(1 \rightarrow a)$. It is easy to check that the non-deterministic morphisms $(1 \rightarrow a) \rightarrow (1 \rightarrow b)$ are automatically positive real matrices of type $a \rightarrow b$ and that the deterministic morphisms of type $(1 \rightarrow a) \rightarrow (1 \rightarrow b)$ are precisely the stochastic matrices of type $a \rightarrow b$. 


Because positive matrices are isomorphic to (bigger) real positive states and positive matrices of matrices are isomorphic to (bigger) real positive matrices through a classical analogue of the CJ-isomorphism, we can adopt a simplified representation. Each state $\phi$ of type $\bigotimes_{i = 1}^n (a_i \rightarrow a_i')$ can be identified with the real positive vector $\phi_{a_1' \dots a_n' a_1 \dots a_n}$ obtained by stacking the columns of the multistochastic matrix $\phi$ on top of one another; whereas each process $T: \bigotimes_{i = 1}^{n} (a_i \rightarrow a_i') \rightarrow  \bigotimes_{k = 1}^m (b_k \rightarrow b_k')$ is identified with a real positive matrix $T_{b_1'\dots b_m' b_1 \dots b_m }^{a_1' \dots a_n' a_1 \dots a_n}$.
The deterministic morphisms can then be taken to be those real positive matrices such that whenever $\phi_{a_1' \dots a_n' a_1 \dots a_n}$ is a representation of a classical non-signalling channel, then  $\psi_{b_1'\dots b_m' b_1 \dots b_m } = T_{b_1'\dots b_m' b_1 \dots b_m }^{a_1' \dots a_n' a_1 \dots a_n}\phi_{a_1' \dots a_n' a_1 \dots a_n}  := T \circ \phi$ is too. This is the definition of the maps of type $ \otimes_i (a_i \rightarrow a_i') \longrightarrow \otimes_k (b_k \rightarrow b_k')$ in the Caus-construction $\mathsf{Caus}(\mat{\reals^+})$ over the category $\mat{\reals^+}$ of positive valued matrices \cite{kissinger_caus}, meaning that we can automatically conclude without direct proof that Boxworld is a symmetric monoidal category. 

\end{example}

\begin{example}[$\mathbf{QBox}$ (the theory of non-signalling quantum boxes)]
Whereas Boxworld as presented above can be understood as the theory of classical channels under a non-signalling tensor product, a theory $\mathbf{QBox}$ can be defined in which systems are quantum channels under a non-signalling tensor product \cite{hefford_decohere_causal}. In the generalised theory $\mathbf{QBox}$, a mild generalisation of it's presentation in \cite{hefford_decohere_causal} to allow for classical interfaces, systems are given by arbitrary terms of the form $(\otimes_i (A_i \rightarrow A_i') , n)$ with $A_i, A_i'$ quantum systems and $n$ a positive integer, and the monoidal product is given by $(\otimes_i (A_i \rightarrow A_i') , n) \otimes_{\mathbf{QBox}} (\otimes_k (B_k \rightarrow B_k') , m): = ((\otimes_i (A_i \rightarrow A_i')) \otimes  (\otimes_k (B_k \rightarrow B_k' ) ) , nm)$. 

We again take $\cat{C}$ to be a category of hybrid quantum-classical maps, this time taking morphisms $(\otimes_i (A_i \rightarrow A_i') , n) \rightarrow (\otimes_k (B_k \rightarrow B_k' ) , m)$ to be families $\{M_x^{y}\}_{x,y} : \otimes_i (A_i^* \otimes A_i' ) \rightarrow  \otimes_k (B_k^* \otimes B_k' ) $ of completely positive maps without any trace-preservation requirement. 
A deterministic morphisms of type $(\otimes_i (A_i \rightarrow A_i') , n) \rightarrow (\otimes_k (B_k \rightarrow B_k' ) , m)$ are taken to be families $\{M_x^{y}\}_{x,y}$ of completely positive maps $M_x^y: \otimes_i (A_i^* \otimes A_i' ) \rightarrow  \otimes_k (B_k^* \otimes B_k' ) $ such that for each fixed setting $x$ the family $\{M_x^{y}\}_{y}$ is a quantum super-instrument, meaning that $\sum_y M^y_x$ is a multi-input quantum superchannel which accepts as input a multipartite non-signalling channel of type $(\otimes_i A_i \rightarrow \otimes A_i')$ and returns as output a multipartite non-signalling channel of type $(\otimes_k B_k \rightarrow \otimes B_k')$. Equivalently, one may simply think of this requirement as being a morphism of type $\otimes_i ((A_i \rightarrow A_i') ) \rightarrow  \otimes_k ((B_k \rightarrow B_k' ) ) $ in higher-order quantum theory \cite{kissinger_caus, Bisio_2019}.
\end{example}

We now present two more advanced examples from the literature based on the CPM construction \cite{abramsky_coecke}.

\begin{example}[Quantum theory over other semirings]\label{ex:cpm}
    If one is willing to relax $\reals^+$ to other semirings of probabilities, then many other examples of generalised theories arise by taking $\cat{C}=\mathsf{Split}^\dag(\mathsf{CPM}(\mat{S}))$ for some commutative semiring $S$ with involution $(-)^*:S\morph{}S$ \cite{gogioso_fantastic}.
    We start with the category $\mat{S}$ whose objects are the natural numbers and whose morphisms are the matrices with entries from $S$.
    Composition is given by matrix multiplication and tensor product by Kronecker product.
    To $\mat{S}$ we apply the CPM-construction which is an abstract method for building a category of completely positive maps from an underlying compact closed one \cite{selinger_idempotents}.
    The objects of $\cpm{\mat{S}}$ are of the form $n\times n$ for a natural number $n$ and the morphisms $n\times n\morph{}m\times m$ take the following form,
    \begin{equation*}
        \input{figs/cpm.tikz}
    \end{equation*}
    for some $f:n\morph{}x\times m$, where $f^*$ is given by applying the involution $(-)^*$ to every entry of the matrix.
    One should think of $\mat{S}$ as the pure quantum-like theory and $\cpm{\mat{S}}$ as the mixed theory given by ``doubling'' the pure one, much like the transition from quantum states to density matrices.
    
    To the mixed quantum-like theory $\cpm{\mat{S}}$ we then add classical-like systems by splitting the $\dag$-idempotents and moving to a new category $\mathsf{Split}^\dag(\mathsf{CPM}(\mat{S}))$ \cite{selinger_idempotents}.
    Avoiding the intricate details of this construction (though the interested reader is strongly recommended to consult \cite{selinger_idempotents,coecke_classicality,gogioso_fantastic}), the idea is that the decoherence maps $\mathrm{dec}:n\times n\morph{} n\times n$ which act to zero-out the off-diagonal elements of a density matrix are $\dag$-idempotents in $\cpm{\mat{S}}$.
    \begin{equation*}
        \input{figs/dec1.tikz} \ = \ \input{figs/dec2.tikz} ::
        \begin{pmatrix} 
            a_{00} & a_{01} & \dots & a_{0,n-1} \\
            a_{10} & a_{11} & \dots & a_{1,n-1} \\
            \vdots & \ddots & \ddots & \vdots \\
            a_{n-1,0} & \dots & \dots & a_{n-1,n-1}
        \end{pmatrix}
        \mapsto
        \begin{pmatrix} 
            a_{00} & 0 & \dots & 0 \\
            0 & a_{11} & \dots & 0 \\
            \vdots & \ddots & \ddots & \vdots \\
            0 & \dots & \dots & a_{n-1,n-1}
        \end{pmatrix}
    \end{equation*}
    In $\mathsf{Split}^\dag(\mathsf{CPM}(\mat{S}))$, these decoherence maps are promoted to objects $(n\times n, \mathrm{dec})$ with a morphism between such objects $f:(n\times n, \mathrm{dec})\morph{} (m\times m, \mathrm{dec})$ being a map $f:n\times n \morph{} m\times m$ which is moreover compatible with the decoherence maps so that $\mathrm{dec} \circ f\circ \mathrm{dec} = f$.
    This forces $f$ to be diagonal in the chosen basis of decoherence.
    As a result, the states of decohered systems are unnormalised probability distributions with values in the sub-semiring $S^+\subseteq S$ of positive elements of $S$ (those elements $z\in S$ of the form $z=\sum_i s_i^*s_i$ for some $s_i\in S$).
    The full subcategory of decohered systems of the form $(n\times n, \mathrm{dec})$ can be shown to be equivalent to $\mat{S^+}$ giving the classical-like sub-theory \cite{gogioso_fantastic}.

    Taking $\cat{C}=\mathsf{Split}^\dag(\mathsf{CPM}(\mat{S}))$, the deterministic subcategory $\cat{C}^{\mathbf{d}}$ is given by taking just those maps of $\cat{C}$ that are normalised under the trace as in the following diagram.
    \begin{equation*}
        \input{figs/cpm_trace1.tikz} \ = \ \input{figs/cpm_trace2.tikz}
    \end{equation*}
    This corresponds with the standard trace on matrices, so that on the classical-like sub-theory this forces the $S^+$-valued probability distributions to be normalised, $\sum_{i=0}^{n-1} a_{i,i} = 1$. 

    To summarise, the categories arrange themselves as,
    \begin{equation*}
        \begin{tikzcd}
           \substack{\stoch_{S^+} \\ \text{normalised } S^+\text{-valued probability} \\ \text{distributions and stochastic maps}} \ar[r] \ar[d] & \cat{C}^{\mathbf{d}} \ar[d] \\
           \substack{\mat{S^+} \\ S^+\text{-valued classical theory}} \ar[r] & \substack{\mathsf{Split}^\dag(\mathsf{CPM}(\mat{S})) \\ S\text{-valued quantum and classical theory}}
        \end{tikzcd}
    \end{equation*}
    This construction yields (up to equivalence): standard quantum theory upon taking $S$ to be $\mathbb{C}$ equipped with complex conjugation; real quantum theory with $S=\reals$ equipped with the trivial involution; and $p$-adic quantum theory upon taking $S=\mathbb{Q}_p(\sqrt{\tau})$ to be a quadratic extension of the $p$-adics equipped with the involution $\sqrt{\tau} \mapsto -\sqrt{\tau}$.
    See \cite{gogioso_fantastic} for even more examples.
\end{example}

\begin{example}[Hypercubic theories]
    Many more CPTs arise as a further generalisation of the previous example by considering $\cat{C}=\mathsf{Split}^\dag(\mathsf{CPM}_G(\mat{S})$ using the higher-order CPM construction for some finite group $G$ \cite{gogioso_cpm}.
    Suitable choices of $G$ and $S$ lead to the theory of density hypercubes \cite{gogioso_hyper,hefford_hypercubes} with $\reals^+$ as its probability semiring, and also to more exotic theories based on Galois extensions $k\subseteq K$ which have a certain closure of the image of the field norm $N_{K/k}$ as their probability semirings \cite{hefford_galois}.
\end{example}

\subsection{Instruments}

Supermaps, and in particular process matrices, rely upon the notion of an \textit{instrument}, a special case of processes representing a non-destructive measurement. These can be defined within any generalised theory.
Recall that in quantum theory, an instrument is a collection $\{M^y\}_y$ of completely positive trace non-increasing (CPTNI) maps, where the $y$ are elements of some classical outcome set $\mathbb{Y}$, and such that the sum of every element of the instrument yields a completely positive trace-preserving (CPTP) map $M$,
\[
    \sum_{y\in\mathbb{Y}} M^y = M\quad \text{is a CPTP map}.
\]
The classical variable $y$, usually called an \textit{outcome}, represents which measurement result was observed by the agent after they have acted with the instrument. 

In the study of causal correlations, we are often interested in the relationship between an agent's choice and another agent's observed outcome. 
In this context, we consider doubly-parameterised families of CPTNI maps $M_x^y$ with the property that, for every possible attribution of the setting $x$, the result is an instrument. That is, for all $x$, the sum over outcomes yields a CPTP map,
\[
\forall x\in\mathbb{X}, \quad \sum_{y\in\mathbb{Y}} M_{x}^{y} \quad \text{is a CPTP map}.
\]
The classical variable $x$, usually called a \textit{setting}, represents the (possibly random) conditioning of which instrument the agent decides to use. 
We call the family $M_x^y$ a \textit{parameterised instrument} and they can be defined in any generalised theory simply as a process of type $M:A \otimes \mathbb{X} \rightarrow B \otimes \mathbb{Y}$, 
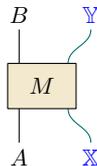
\[  \input{figs/param_inst.tikz}  \] 

The components of the parameterised instrument are recovered by evaluating $M$ at the points of the classical distributions, $x\in\mathbb{X}$ and $y\in\mathbb{Y}$.
Each $x$ corresponds to a state $\left| x \right): I\morph{}\mathbb{X}$ and each $y$ to an effect $\left( y \right| : \mathbb{Y}\morph{}I$, represented respectively as processes with no input and no output.
Then we have \[  M_{x}^y \ = \ \input{figs/param_inst2.tikz}.  \]

\subsection{Instruments in Boxworld}

In most of the generalised theories outlined above, the structure of instruments is rather straightforward, as it is in finite-dimensional quantum theory. However, in Boxworld, the causal constraints of the theory make the structure of instruments more subtle. This represents a point of departure between our presentation of Boxworld and that given in \cite{bavaresco2024indefinitecausalorderboxworld}.

In general, following the recipe above, an instrument in the generalised theory Boxworld is a morphism of type $(A \rightarrow A') \otimes \mathbb{X} \rightarrow (B \rightarrow B') \otimes \mathbb{X}'$. We can represent such a morphism graphically as follows:
\[
\input{figs/ns_tensor_def_a.tikz}
\ : \ (A \rightarrow A') \otimes \mathbb{X} \rightarrow (B \rightarrow B') \otimes \mathbb{X}'.
\]

Alternatively, using the derivable isomorphisms of the BV logic of higher-order causal categories \cite{SimmonsKissinger2022}, together with a few specific features of higher-order causal categories (HOCCs), we obtain:
\begin{align}
& (A \rightarrow A') \otimes \mathbb{X} \rightarrow (B \rightarrow B') \otimes \mathbb{X}'  \\ 
\cong {} & (A \rightarrow A') \rightarrow (\mathbb{X} \otimes ((B \rightarrow B') \otimes \mathbb{X}')^*)^*  \\
\cong {} & (A \rightarrow A') \rightarrow (\mathbb{X} \otimes (\mathbb{X}' \seq (B \rightarrow B'))^*)^*  \\
\cong {} & (A \rightarrow A') \rightarrow (\mathbb{X} \seq (\mathbb{X}' \seq (B \rightarrow B'))^*)^*  \\
\cong {} & (A \rightarrow A') \rightarrow (\mathbb{X}^* \seq (\mathbb{X}' \seq (B \rightarrow B')))  \\
\cong {} & (A \rightarrow A') \rightarrow ((\mathbb{X}^* \seq \mathbb{X}') \seq (B \rightarrow B'))  \\
\cong {} & (A \rightarrow A') \rightarrow ((\mathbb{X} \rightarrow \mathbb{X}') \seq (B \rightarrow B')) .
\end{align}
This can be depicted graphically as
\[
\input{figs/ns_tensor_def_b.tikz}
\ : \ (A \rightarrow A') \rightarrow ((\mathbb{X} \rightarrow \mathbb{X}') \seq (B \rightarrow B')).
\]

In short, in Boxworld, an instrument is a classical comb subject to a no-signalling constraint from $B$ to $\mathbb{X}'$. This contrasts with the instrument set specified for Boxworld in \cite{bavaresco2024indefinitecausalorderboxworld}, which allows \textit{any} comb of type $(A \rightarrow A') \rightarrow (B \otimes \mathbb{X} \rightarrow B' \otimes \mathbb{X}')$ to function as an instrument. We adopt the terminology of \cite{bavaresco2024indefinitecausalorderboxworld} and refer to such maps as \textit{probabilistic operations}. Graphically, the \textit{probabilistic operations} of \cite{bavaresco2024indefinitecausalorderboxworld} can be depicted simply as
\[
\input{figs/ns_tensor_def_c.tikz}
\ : \ (A \rightarrow A') \rightarrow (B \otimes \mathbb{X} \rightarrow B' \otimes \mathbb{X}').
\]

From now on, we refer to instruments in the generalised theory Boxworld as \textit{Boxworld instruments}, in contrast to \textit{probabilistic operations}.

There is a physical motivation for abandoning \textit{probabilistic operations} and focusing exclusively on \textit{Boxworld instruments}. Because the system $B$ forms part of the output system in Boxworld, allowing signalling from $B$ to $\mathbb{X}'$ would permit superluminal signalling mediated by a state. Since we wish to regard Boxworld as a no-signalling theory, it is therefore natural to impose this constraint on its instruments. It is worth noting, however, that this is not an additional constraint that we have chosen to impose on Boxworld. Rather, it arises naturally by identifying precisely where the classical subtheory lies within Boxworld, using the no-signalling tensor product, and then determining the precise form of the maps in the theory of type $(A \rightarrow A') \otimes \mathbb{X} \rightarrow (B \rightarrow B') \otimes \mathbb{X}'$, since it is these maps that are most naturally interpreted as instruments in any generalised theory.

\section{Supermaps on Generalised Theories with Channel-State Duality}\label{sec:GTwDuality}

So far, we have established what we mean by a general physical theory: a theory in which causal classical information flow, and in which probabilistic computations are embedded. The theories described are theories of processes, although with particular special processes (those of type $\rho : I \rightarrow A$ for some system $A$) that represent states. In this setting, we naturally think of processes as things that transform states. A process of type $f: A \rightarrow B $ transforms states of type $\rho : I \rightarrow A$ to states of type $ \rho' : I \rightarrow A'$ via post-composition, i.e.\ by $f(\rho) : = f \circ \rho$. 

On the other hand, supermaps and process matrices are concerned not with transformations of states, but transformations of the processes themselves.
Typically, supermaps are depicted with the following pictures, for one-input and multi-input cases: \[  \input{figs/new_context_1.tikz} \quad , \quad \quad \input{figs/new_context_2.tikz} .  \] 

In the one-input supermap diagram above, if the wires are labelled from bottom to top and then from left to right, $\phi$ is a morphism of type $(C \otimes E) \rightarrow (D \otimes F)$, whereas the supermap has a type $(C \rightarrow D) \rightarrow (A \rightarrow B)$. As we mentioned before, this latter type does not refer to a valid morphism a priori. 
This, of course, demands an answer to the following question: where do such pictures formally live then?
Are they already in one of the available categories, and just remain to be identified? 
Or are they additional structures that need to be specified?
In the following, we will demonstrate how the categorical supermap approach gives a simple and stable answer to this question, and so in turn gives a methodical framework for the study of indefinite causal order in any general physical theory (recall that supermaps are intrinsically linked with higher-order maps see e.g. \cite{chiribella_switch,kissinger_caus,Bisio_2019,hoffreumon2022projective}). 
In this section we will first outline the more traditional way of thinking about supermaps and about higher-order probabilistic correlations.

Supermaps were first considered in the context of quantum theory. A quantum supermap is a map which sends channels to channels rather than states to states. More precisely, let us note $\mathcal{L}(A,B)$ the space of linear maps between spaces $A$ and $B$, and $\mathcal{L}(A,A):= \mathcal{L}(A)$ the space of linear operators on $A$, then a quantum channel $\mathcal{E}: A \rightarrow A'$ between Hilbert spaces $A$ and $A'$ is a linear map $\mathcal{E} \in \mathcal{L}(\mathcal{L}(A),\mathcal{L}(A'))$ which is completely positive in the sense that whenever $\rho \geq 0$ then $\mathcal{E} \otimes \mathcal{I} (\rho) \geq 0$ and furthermore trace preserving in the sense that $\tr[\mathcal{E}(\rho)] = \tr[\rho]$. In accordance, a quantum supermap (or quantum superchannel)  $\mathcal{S}: (A\rightarrow A')\rightarrow (B \rightarrow B')$ is a linear map $\mathcal{S} \in \mathcal{L}(\mathcal{L}(\mathcal{L}(A),\mathcal{L}(A')) , \mathcal{L}(\mathcal{L}(B),\mathcal{L}(B')))$ such that whenever $\mathcal{E} \in \mathcal{L}(\mathcal{L}(A \otimes X),\mathcal{L}(A' \otimes X' )$ is a bipartite quantum channel of type $(A \otimes X) \rightarrow (A' \otimes X')$, then $\mathcal{S} \otimes \mathcal{I} (\mathcal{E})$ is a bipartite quantum channel of type $(B \otimes X) \rightarrow (B' \otimes X')$ and this relation must hold for any $X, X'$.  
The intended graphical interpretation of a quantum supermap is the following \[  \mathcal{S} \otimes \mathcal{I} (\mathcal{E}) \ = \  \input{figs/new_context_1b.tikz}. \]
This definition can be generalised to the multi-input setting. A multi-input quantum supermap is a linear map such that for any sequence $\mathcal{E}_i$ of quantum channels,  $\mathcal{S} \otimes \mathcal{I} (\mathcal{E}_1 \otimes \dots \otimes \mathcal{E}_n)$ is a quantum channel. This time the intended graphical interpetation is the following \[ \mathcal{S} \otimes \mathcal{I} (\mathcal{E}_1 \otimes \dots \otimes \mathcal{E}_n)\ = \  \input{figs/new_context_2b.tikz}. \]

Supermaps with a single input admit a concrete representation quite similar in spirit to Stinespring's dilation theorem (see e.g. \cite{Wilde_2013}, \S 5.2.2), which recognises that every quantum channel can be concretely realised as an isometry plus discarding on the environment. In the higher-order setting, the analogous theorem states that every supermap admits a concrete realisation in terms of a quantum comb \cite{chiribella_networks}, that is, a quantum circuit into which a hole has been cut. This realisation means that the supermap $\mathcal{S}$ can be split into two channels $\mathcal{S} = \mathcal{N}_u \circ \mathcal{N}_l$ of respective types $\mathcal{N}_l: B \rightarrow (E \otimes A)$ and $\mathcal{N}_u: (E \otimes A') \rightarrow B'$, and which respectively represent the quantum circuit before and after the hole; this is graphically interpreted like \[  \input{figs/new_context_1b.tikz}    \ = \  \input{figs/comb_def_1.tikz} . \]
On the other hand, for multi-input supermaps such concrete realisations cannot be given without imposing additional causality constraints, as witnessed for instance by those multi-inputs which exhibit indefinite causal order \cite{chiribella_switch}.

Fundamental to the above-mentioned way of defining higher-order maps is the existence of the Choi-Jamiolkowski (CJ) isomorphism \cite{Jamiolkowski_1972,Choi_1975}. In particular, to form the expression $\mathcal{S} \otimes \mathcal{I} (\mathcal{E})$ requires repeated use of the property that $\mathcal{L}(\mathcal{L}(A),\mathcal{L}(B)) \cong \mathcal{L}(A)^{*} \otimes \mathcal{L}(B)$ (where $\mathcal{L}(A)^*$ is the algebraic dual of $\mathcal{L}(A)$); this property of finite-dimensional vector spaces acts as a precursor to the CJ-isomorphism. 
Before we introduce categorical supermaps in full generality, let us first outline a concrete approach to defining supermaps on generalised theories in the presence of such a channel-state duality.
This definition is put forward in  \cite{wilson_locality, wilson_polycategories}; it extends the definition of \cite{kissinger_caus} to ensure the validity of the partial application of higher-order maps. In the special case of pre-causal categories as considered in \cite{kissinger_caus}, explicit reference to partial application is not necessary so long as complete positivity is assumed from the outset. 
As mentioned already, the abstraction of channel-state duality to the language of category theory is known as compact closure.
\begin{definition}[Compact Closed Category]
    Let $\cat{C}$ be a symmetric monoidal category.
    $\cat{C}$ is compact closed if for every object $A$ there exists a dual object $A^*$, a state $\cup: I \rightarrow A \otimes A^{*}$ and an effect $\cap: A^{*} \otimes A \rightarrow I$ such that 
    \begin{equation}
    	(\cap \otimes 1) \circ (1 \otimes \cup) = 1 \:,
    \end{equation} 
	where $1: A \rightarrow A$ is the identity morphism. I.e., such that the following diagram holds
    \[\input{figs/snake_1.tikz} \ = \ \input{figs/snake_2.tikz}.  \]
    
\end{definition}

Given an embedding $\cat{C} \hookrightarrow \cat{D}$ of a symmetric monoidal category $\cat{C}$ into a compact closed one $\cat{D}$, one can always define a candidate notion of higher-order processes using the channel-state duality of $\cat{D}$.
We refer the reader to \cite{wilson_locality, wilson_polycategories} for a general discussion on this construction. The key point for now is that the discussion of this section is focused on generalised theories for which the embedding of the deterministic processes into the non-deterministic ones $\cat{C}^{\mathbf{d}} \hookrightarrow \cat{C}$ is an embedding of a symmetric monoidal category into a compact closed one. Think of these theories as those in which one cannot always achieve the cap and cup deterministically, but can still obtain them with a non-zero chance. Quantum theory is a typical instance of such a theory: knowing that the cup is realised by a Bell state (or a maximally entangled state for higher dimensions), and that the cap is realised by observing the system in this state, it is indeed clear that one can prepare a given Bell state with a probability of one, but that one cannot measure every system in this exact Bell state with a probability of one (without pre-processing the system by a CPTP map or without postselecting the outcome); every quantum state (but the three Bell states orthogonal to it) has a non-zero chance to be measured in this Bell state. So while the cap does not belong to the deterministic effects of quantum theory, it is still a valid non-deterministic effect. 

We will refer to these generalised theories whose instruments enjoy a channel-state duality as the \textit{generalised theories with channel-state duality}; they will be the focus of this section.

\begin{definition}[Generalised Theories with Channel-State Duality]
    A generalised theory has channel-state duality if its non-deterministic theory $\cat{C}$ is a compact closed category such that:
    \begin{itemize}
        \item The cup $\cup$ is an instrument element up to a scalar \[  \input{figs/cup.tikz} \ = \ \input{figs/cup_instrument.tikz}\]
        \item The cap $\cap$ is an instrument element up to a scalar  \[   \input{figs/cap.tikz} \ = \ \input{figs/cap_instrument.tikz}\]
    \end{itemize}
\end{definition}


\begin{example}[Classical theory]
Classical theory has channel state duality. The non-deterministic theory $\mat{\reals^+}$ is compact closed, with cup given by the perfectly correlated state $p_{xy}$ which is equal to $1$ when $x = y$ and $0$ otherwise.
The cap is given by the perfectly correlated effect $N^{xy}$ which is equal to $1$ when $x = y$ and $0$ otherwise. The cup can be normalised to a stochastic matrix directly by taking $\frac{1}{d} p_{xy}$ where $d$ is the dimension of the system at hand. 
The cap can be completed to the unique normalised effect (the discard) defined by $E^{xy}$ which is $1$ for every pair $x,y$. The completion is given by adding the effect $Q^{xy} = E^{xy} - N^{xy}$ which is everywhere non-negative, and in fact is $1$ whenever $x \neq y$ and $0$ otherwise. Consequently, we may define the measurement $M : \mathbb{X} \otimes \mathbb{Y} \rightarrow \mathbb{Z}$ where $\mathbb{Z}$ is the $2$-dimensional system by $M_{xy}^{z}$ equals $1$ when $x = y$ and $z = 0$ or when $x \neq y$ and $z = 1$. This is a normalised channel since discarding $\mathbb{Z}$ returns $E^{x,y}$, and furthermore it has a component (the $0^{th}$-component) which is the cap up to a scalar (in this case the scalar is $1$).  
\end{example}

\begin{example}[Quantum theory]
    To see that quantum theory has channel-state duality firstly note that the non-deterministic theory consists of the category of completely positive maps which is compact closed.
    The cap on an $n$ dimensional system is, up to a scalar, the effect $\bra{\psi} := \frac{1}{\sqrt{n}}\sum_{i=0}^{n-1} \bra{ii}$ which exists as an instrument component by completing it to a PVM: just take the orthogonal projection $P := \ket{\psi}\bra{\psi}$ and note that $1-P$ is also an orthogonal projection forming the other component of the PVM.
\end{example}

\begin{example}[Boxworld]
    To see that Boxworld has channel-state duality, first note that its non-deterministic category is $\mat{\reals^+}$ which is indeed compact closed.
    Next we must check that there exist instruments whose elements are the caps and cups of $\mat{\reals^+}$.
    This is the requirement that there be an instrument $M: I \rightarrow  (A \rightarrow A') \otimes (A' \rightarrow A) \otimes  \mathbb{X}$ with some component $M^x$ given by the swap channel $\texttt{Swap}: A \otimes A'\rightarrow A'\otimes A$ up to a scalar. This follows immediately since any permutation can be completed by convex combination with another complement channel ${\texttt{Swap}}'$ to produce the completely depolarising channel, which is indeed a non-signalling channel. Consequently we can define the instrument $M = p \texttt{Swap} \otimes |0) + (1-p) { \texttt{Swap} }' \otimes |1)$ which is indeed of the type $(A \rightarrow A') \otimes (A' \rightarrow A) \otimes  \mathbb{X}$\footnote{it is easy to check that this instrument is non-signalling from any $A$ or $A'$ input to the $\mathbb{X}$ output.}. For the same reasons there exists an instrument with component given by the cap up to a scalar, one again finds the effect which completes the cap (up to a scalar) to the discard of classical theory, and then constructs the instrument which in one outcome branch implements $p \cap \otimes |0) $ and in $|1)$ implements the effect that completes the cap to the discard with coefficient $(1-p)$. 
\end{example}

\begin{example}[Quantum boxes]
    To see that $\mathbf{QBox}$ has channel-state duality, first note that its non-deterministic category is again a theory of completely positive maps which is indeed compact closed. Interrpeted as a state in $\mathbf{QBox}$ the cup is (just as in boxworld) is the swap (quantum) channel, and so to check channel-state duality we must check that $M: I \rightarrow  (A \rightarrow A') \otimes (A' \rightarrow A) \otimes  \mathbb{X}$ with some component $M^x$ given by the quantum swap channel $\texttt{Swap}: A \otimes A'\rightarrow A'\otimes A$ up to a scalar. This follows immediately since the difference between the completely depolarising channel and $\alpha U$ for any unitary $U: H \rightarrow H$ and scalar $0 \le \alpha \le \frac{1}{d^2}$ with $d$ the dimension of the Hilbert space $H$, is a completely positive map. Indeed, recalling the form of the Choi matrix $J(\Phi)=\sum_{i,j}|i\rangle\langle j|\otimes \Phi(|i\rangle\langle j|)$, the Choi-matrix for the completely depolarising channel is given by $J(\mathcal D)=\frac1d\, I\otimes I$ and the Choi matrix for the channel corresponding to $U$ is given by $J(\mathcal U)=\sum_{i,j}|i\rangle\langle j|\otimes U|i\rangle\langle j|U^\dagger=|\Omega_U\rangle\langle \Omega_U| $, where $|\Omega_U\rangle:=\sum_i |i\rangle\otimes U|i\rangle$.  For the Choi-matrix of the difference, we see that $J(\Phi)=\frac1d\,I\otimes I-\alpha |\Omega_U\rangle\langle \Omega_U| $ has eigenvalue $\frac1d-\alpha d $
on $|\Omega_U\rangle$, and eigenvalue $1/d$ on its orthogonal complement. Therefore $J(\Phi)\ge 0 \Longleftrightarrow \frac1d-\alpha d\ge 0 \Longleftrightarrow \alpha\le \frac1{d^2} $.
 \end{example}

\begin{example}[Real quantum theory]
    For much the same reason as standard quantum theory, real quantum theory also has channel-state duality.
    Its non-deterministic category $\mathsf{Split}^\dag(\cpm{\mat{\reals}})$ is compact closed \cite{selinger_idempotents}.
    Like in complex quantum theory, the projection $P$ onto the cap can be completed to a PVM, and all we have to do is check that the required component $1 - P$ contains only real entries in its associated matrix.
    This is immediate upon inspection.
\end{example}

\begin{example}[Quantum theory over other semirings]\label{ex:cs_duality_semirings}
    The previous example can be taken further by considering the more general class of theories described in Example \ref{ex:cpm}.
    While the non-deterministic category $\mathsf{Split}^\dag(\cpm{\mat{S}})$ is always compact closed, it is not always the case that the cap can be completed to an instrument, depending crucially on the choice of $S$.
    Restricting $S$ to be a field of characteristic zero, a sufficient condition for the generalised theory to have channel-state duality is for $S$ to be any field containing all square roots, that is, any field extension of the constructable numbers.
    In this case, $1-P$ has all its entries in $S$ for every dimension of system.
    Note that this case includes $S=\mathbb{C}$ and $\reals$, and also the constructable numbers, the algebraic numbers $\overline{\mathbb{Q}}$ and the algebraic completion of the $p$-adic numbers $\overline{\mathbb{Q}_p}$ (equipping the latter with involution $\sqrt{-1}\mapsto -\sqrt{-1}$).
\end{example}

\begin{example}[Non-examples]
    An example of a theory \textit{without} channel-state duality is infinite dimensional quantum theory, since its deterministic theory is the category $\mathsf{Hilb}$ of all Hilbert spaces and bounded linear maps which is not compact closed.
    Further examples of generalised theories without channel-state duality, can be found by suitably picking $S$ in Example \ref{ex:cs_duality_semirings}.
    For instance, in quantum theory over $\mathbb{Q}(\sqrt{2})$ the cup $\ket{00}+\ket{11}+\ket{22}$ of a 3-dimensional system cannot be normalised because $\sqrt{3}\notin \mathbb{Q}(\sqrt{2})$.
    This prevents the cup from being a physical state for all dimensions, let alone being a component of an instrument.
\end{example}


\begin{definition}[Single-input CJ-supermap]
    A CJ-supermap of type \[ (A \Rightarrow A') \rightarrow (B \Rightarrow B') \] on a generalised theory with channel-state duality is a process of type $S: B \otimes A' \rightarrow B' \otimes A$ in $\cat{C}$ such that for every $\phi : A \otimes X \rightarrow A' \otimes X'$ in $\cat{C}^{\mathbf{d}}$ the following map is also in $\cat{C}^{\mathbf{d}} $. \[ \input{figs/cj_supermap_1.tikz} \]
\end{definition}

In other words, a CJ-supermap on a generalised theory is a $\mathcal{C}$-supermap on $\cat{C}^{\mathbf{d}}$ as defined in  \cite{wilson_locality, wilson_polycategories}. 
This definition is easily generalised to the multi-input case.

\begin{definition}[Multi-input CJ-supermap]
    A CJ-supermap of type \[ \bigtimes_{i}(A_i \Rightarrow A_i') \rightarrow (B \Rightarrow B') \] on a generalised theory is a process of type $S: B \otimes (A_1' \otimes \dots \otimes A_n') \rightarrow B' \otimes (A_1 \otimes \dots \otimes A_n)$ in $\cat{C}$ such that for every collection $\phi_i : A_i \otimes X_i \rightarrow A_i' \otimes X_i'$ of processes in $\cat{C}^{\mathbf{d}} $ the following process is also in $\cat{C}^{\mathbf{d}} $. \[ \input{figs/cj_supermap_2.tikz} \] 
\end{definition}

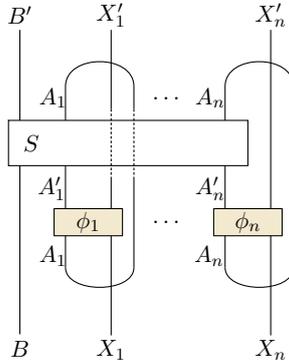
\begin{theorem}
The following equalities hold:
    \begin{itemize}
        \item Classical Supermaps = CJ supermaps on classical theory
        \item Quantum Supermaps = CJ supermaps on quantum theory
    \end{itemize}
\end{theorem}
\begin{proof}
A direct proof is given in \cite{wilson_locality}.
\end{proof}

\subsection{Supermaps on Boxworld}
Since the introduction of a general framework for supermaps on arbitrary physical theories, 
more direct attempts have been made to define supermaps in general physical theories. The most closely related is the proposal for Boxworld \cite{bavaresco2024indefinitecausalorderboxworld}.
In that setting, supermaps were defined as a particular class of process tensors that are also required to have the \textit{no signalling without system exchange} (NSWSE) property.
To understand how this definition departs from the categorical approach, we must first introduce the preliminary notion of a non-signalling probabilistic operation. 
 
\begin{definition}
A probabilistic operation $\{T_{AA'BB'}\}_{X}^{X'}$ is termed non-signalling (NS) if, for every non-signalling distribution $\rho_{YY'AA'}$, the contraction $\sum_{x'} T \star_{AA'} \rho$ is independent of $X$. In purely graphical language, we say that a morphism in $\mat{\reals^+}$ depicted by \[  \input{figs/ns_tensor_def_0.tikz} , \]
is non-signalling if for every non-signalling channel $\rho$ the following equality holds \[ \input{figs/ns_tensor_def_1.tikz}   \ = \    \input{figs/ns_tensor_def_2.tikz} .  \]
\end{definition}

In the approach of \cite{bavaresco2024indefinitecausalorderboxworld}, the legitimate higher-order maps on Boxworld are taken to be those that satisfy \textit{no signalling without system exchange}. That is, for any pair ${T_1}_{x_1}^{y_1},{T_2}_{x_2}^{y_2}$ of NS probabilistic operations in the above sense, the result $p(y_1, y_2 \vert x_1, x_2) = W \star ({T_1}_{x_1}^{y_1}, {T_2}_{x_2}^{y_2})$ obtained by contracting them with the tensor $W$ must be a non-signalling channel. Graphically, when $T_1$ and $T_2$ are non-signalling, the following must be a non-signalling channel: \[  \input{figs/box_explain_0.tikz}. \]

It is important to note, however, that in the generalised theory Boxworld, the instruments already satisfy another no-signalling condition: no signalling from $B$ to the outcome set $Y_1$ and no signalling from $D$ to the outcome set $Y_2$.
We will refer to probabilistic operations that are both Boxworld instruments and non-signalling as \textit{non-signalling Boxworld instruments} (NS-Boxworld instruments).
\begin{theorem}
The NS-Boxworld instruments are the morphisms of type \[ \input{figs/nswse.tikz}, \] in the monoidal category Boxworld. In other words, an NS-Boxworld instrument is a higher-order classical process of type $T: (A \rightarrow A') \rightarrow ((B \rightarrow B') \otimes (X \rightarrow X'))$.
\end{theorem}
\begin{proof}
Note that the NS condition can be reinterpreted within the circuit-theoretic language of Boxworld. A process in Boxworld is NS if it satisfies the following condition: \[ \forall \rho,  x, x': \ \input{figs/nswse_cond_1.tikz} \ = \ \input{figs/nswse_cond_2.tikz}. \]
    When checking the NSWSE condition, we are checking the influence of $X$ on the rest of the state in the following diagram: \[  \input{figs/nswse_scenario_2.tikz}, \] which indeed satisfies no-signalling by Theorem 4 of \cite{Wilson_causal}\footnote{That theorem states that, in any closed symmetric monoidal category with ``no correlations with single-state objects,'' there is no signalling via post-selection on deterministic effects of type $(\mathbb{X} \rightarrow \mathbb{X}') \rightarrow I$. The higher-order classical theory $\mathbf{Caus}[\mathbf{Mat}[\mathbb{R}^{+}]]$ is a closed symmetric monoidal category, is easily seen to satisfy the no-correlation condition, and contains Boxworld via a fully faithful, strong monoidal embedding.}. 
\end{proof}

From now on, we will refer to CJ-supermaps on Boxworld as \textit{Boxworld supermaps}. 
Boxworld supermaps are not precisely the NSWSE tensors because they impose non-signalling outputs with respect to a smaller class of instruments. To clarify the distinction, let us state the definition of NSWSE tensors precisely.

\begin{definition}
We say that a classical tensor $W$ is NSWSE if 
\begin{itemize}
\item for any pair of probabilistic operations ${T_1}_{x_1}^{y_1},{T_2}_{x_2}^{y_2}$, the contraction $p(y_1, y_2 \vert x_1, x_2) = W \star ({T_1}_{x_1}^{y_1}, {T_2}_{x_2}^{y_2})$ is a stochastic channel (a normalised conditional probability distribution),
\item for any pair of NS probabilistic operations ${T_1}_{x_1}^{y_1},{T_2}_{x_2}^{y_2}$, the contraction $p(y_1, y_2 \vert x_1, x_2) = W \star ({T_1}_{x_1}^{y_1}, {T_2}_{x_2}^{y_2})$ is a non-signalling channel.
\end{itemize}
\end{definition}
We will find that Boxworld supermaps do not satisfy NSWSE. However, they do satisfy a weaker version based directly on the restriction from probabilistic operations to Boxworld instruments. 
\begin{definition}
We say that a classical tensor $W$ is weakly-NSWSE (w-NSWSE) if 
\begin{itemize}
\item for any pair of Boxworld instruments ${T_1}_{x_1}^{y_1},{T_2}_{x_2}^{y_2}$, the contraction $p(y_1, y_2 \vert x_1, x_2) = W \star ({T_1}_{x_1}^{y_1}, {T_2}_{x_2}^{y_2})$ is a stochastic channel (a normalised conditional probability distribution),
\item for any pair of NS-Boxworld instruments ${T_1}_{x_1}^{y_1},{T_2}_{x_2}^{y_2}$, the contraction $p(y_1, y_2 \vert x_1, x_2) = W \star ({T_1}_{x_1}^{y_1}, {T_2}_{x_2}^{y_2})$ is a non-signalling channel.
\end{itemize}
\end{definition}
In fact, we will find that Boxworld supermaps are precisely characterised by complete satisfaction of the w-NSWSE condition. To make the notion of completeness precise, let us first define multipartite non-signalling instruments.
\begin{definition}
A multipartite non-signalling Boxworld instrument $T_{x_1 \dots x_n}^{y_1 \dots y_n}$ is a higher-order classical process of type \[ T : (A \rightarrow A') \rightarrow ((B \rightarrow B') \otimes_i (\mathbb{X}_i \rightarrow \mathbb{Y}_i)), \] that is, a Boxworld instrument that, when applied to a classical channel of type $A \rightarrow A'$, returns a multipartite non-signalling channel $B \otimes X_1 \otimes \dots \otimes X_n \rightarrow B' \otimes Y_1 \otimes \dots \otimes Y_n$.
\end{definition}
Here, by multipartite non-signalling we mean that $B$ may only signal to $B'$ and each $\mathbb{X}_i$ may only signal to the corresponding $\mathbb{Y}_i$.
\begin{definition}
A classical tensor $W$ is completely-w-NSWSE (cw-NSWSE) if 
\begin{itemize}
\item for any pair of Boxworld instruments ${T_1}_{x_1}^{y_1},{T_2}_{x_2}^{y_2}$, the contraction $p(y_1, y_2 \vert x_1, x_2) = W \star ({T_1}_{x_1}^{y_1}, {T_2}_{x_2}^{y_2})$ is a stochastic channel (a normalised conditional probability distribution),
\item for any pair of multipartite non-signalling instruments ${T_1}_{x_1^1, \dots, x_1^{n_1}}^{y_1^1, \dots, y_1^{n_1}},{T_2}_{x_2^1, \dots, x_2^{n_2}}^{y_2^1, \dots, y_2^{n_2}}$, the contraction \[p(y_1^1, \dots, y_1^{n_1},  y_2^1, \dots, y_2^{n_2} \vert  x_1^1, \dots, x_1^{n_1},  x_2^1, \dots, x_2^{n_2}) = W \star ({T_1}_{x_1^1, \dots, x_1^{n_1}}^{y_1^1, \dots, y_1^{n_1}},{T_2}_{x_2^1, \dots, x_2^{n_2}}^{y_2^1, \dots, y_2^{n_2}})\] is a multipartite non-signalling channel.
\end{itemize}
Here, by a multipartite non-signalling channel, we mean that each $\mathbb{X}_i^j$ may only signal to the corresponding $\mathbb{Y}_i^j$.
\end{definition}
We will henceforth refer to CJ-supermaps on Boxworld as \textit{Boxworld supermaps}. 
We note that the weak-NSWSE property follows simply by reading off the following diagram \[  \input{figs/boxworld_super_2.tikz}. \] Indeed, by their types each $T_i$ is by definition an NS-Boxworld instrument. 
Furthermore, in precisely this setup, what remains after contracting with $W$ lies in the Boxworld tensor product of two parties; that is, it is a non-signalling channel. Consequently, any Boxworld supermap satisfies weak no signalling without system exchange. In fact, Boxworld supermaps correspond directly to cw-NSWSE tensors. 
\begin{theorem}
Boxworld supermaps are precisely the cw-NSWSE Boxworld tensors.
\end{theorem}
\begin{proof}
First, consider a diagram of the following form: \[  \input{figs/boxworld_super_4a.tikz} \ = \  \input{figs/boxworld_super_4b.tikz}   \] Here, $W$ is a classical tensor that we have rewritten as a string diagram for $\mat{\reals^+}$, using the isomorphism $X \rightarrow Y \cong X^{*} \otimes Y$ to refer explicitly to each individual wire. To conclude that this is a Boxworld supermap, we must check that the diagram is a Boxworld process. Equivalently, we must check that contracting it with a tuple of channels yields a multipartite non-signalling channel: \[ \input{figs/boxworld_super_5.tikz} . \] Thus, for all deterministic Boxworld processes $T_i$ of the following types, we need only check that the result is a multipartite non-signalling channel: \[  \input{figs/boxworld_super_5b.tikz}.  \] This is precisely the cw-NSWSE condition. 
\end{proof}

\begin{theorem}
NSWSE Boxworld tensors $\subseteq$ Boxworld supermaps $\subseteq$ w-NSWSE Boxworld tensors
\end{theorem}
\begin{proof}
Given in the Appendix.
\end{proof}

Given the variety of available conditions, one might naturally wonder about the arguments in favour of each. Here, we argue briefly that Boxworld supermaps provide a natural and suitable definition of higher-order processes on Boxworld, combining several useful properties. We have already argued that the appropriate notion of a Boxworld instrument should incorporate no signalling from $B$ to $X'$. 

At a more heuristic level, a surprising feature of NSWSE Boxworld tensors is that they do not support a simple sequential-composition supermap.
More precisely, the following tensor $W$ is not a valid NSWSE Boxworld tensor, as witnessed by the signalling produced by the pair $T_1,T_2$:
 \[  \input{figs/box_explain_1.tikz}. \]

The same assignment is, however, a valid Boxworld supermap and is therefore also w-NSWSE. This follows directly from the fact that, in any generalised theory with channel-state duality, sequential composition is a CJ-supermap. Moreover, the non-signalling probabilitic operations that witness the failure of NSWSE for the sequential-composition map are not both boxworld instruments (noting the signalling permitted from $D$ to $Y_2$) and therefore cannot witness a failure of w-NSWSE. 

More generally, when NS-Boxworld instruments are plugged into $W$, their additional non-signalling constraints prevent signalling through the sequential-composition map: \[  \input{figs/box_explain_2.tikz}. \]

Thus, while the NSWSE condition is too strong to admit a sequential-composition supermap, the weaker w-NSWSE and cw-NSWSE conditions satisfied by Boxworld supermaps do permit the sequential composition supermaps (referred to in \cite{bavaresco2024indefinitecausalorderboxworld} as the perfect gbit channel). We further conjecture that Boxworld supermaps are precisely the Boxworld tensors satisfying the weak NSWSE condition, in other words, that the w-NSWSE and cw-NSWSE conditions are equivalent. 
It is however, unclear at this stage, whether the spaces of possible non-causal correlations achievable through NSWSE Boxworld tensors and NS probabilistic operations verses cw-NSWSE Boxworld tensors and NS Boxworld instruments coincide. 

\section{Categorical Supermaps}\label{sec:CSM}

In this section, we turn our attention to defining higher-order processes over arbitrary physical theories.
There is no reason to expect all physical theories to come equipped with a channel-state isomorphism and in this case the methods of the last section fail.
It is no longer clear how to give a definition of supermaps and this is the question we now address.

The definition we will use was introduced in \cite{wilson_locality} as a \textit{locally-applicable transformation}, a relaxation of the axioms for traces in monoidal categories \cite{Joyal_Street_Verity_1996}. We will from now on refer to such transformations simply as \textit{categorical supermaps}.
A categorical supermap consists of a family of functions \[  S_{X,X'} :   \cat{C}^{\mathbf{d}} (A \otimes X, A' \otimes X') \rightarrow \cat{C}^{\mathbf{d}} (B \otimes X , B' \otimes X') ,  \] which commute with combs from $\cat{C}^{\mathbf{d}}$ on the environment wires $X,X'$.
Concretely this means that \[  S_{X,X'} ((1_{A'} \otimes g) \circ (\phi\otimes 1_Z) \circ (1_A \otimes f)) = (1_{B'} \otimes g) \circ (S_{YY'}(\phi) \otimes 1_Z) \circ (1_B \otimes f),  \] for any $f:X\morph{}Y\otimes Z$, $g:Y'\otimes Z\morph{}X'$ and $\phi:A\otimes Y\morph{}A'\otimes Y'$.
Note that $S$ refers only to the deterministic physical theory so that any categorical supermap must send deterministic processes to deterministic processes.

We will now begin to frame this definition diagrammatically, to have a formal representation that more closely resembles the intuitive picture for supermaps, and to neatly phrase the generalisation to multiple inputs.
Let us denote the action of the function $S_{X,X'}$ by  \[   \input{figs/baby_lot_1.tikz}    \]
Categorical supermaps are families of functions of the above type such that for every pair of morphisms $f,g$ \[  \input{figs/baby_lot_2.tikz}   \ = \      \input{figs/baby_lot_3.tikz} .   \]
The principle is simple: for $S_{X,X'}$ to model the concept of a supermap that could be applied locally on the left-hand systems $A,A'$, it must commute with any deterministic events that occur on the environment $X,X'$.
The notion above can be extended to a more evocative notation which more closely mirrors the intuitive concept being modelled. This notation will be referred to as box notation for categorical supermaps and has been previously used to model commutativity in \cite{wilson_locality, wilson2023originlinearityunitarityquantum, pinzani2019categoricalsemanticstimetravel}.
In box notation we write the action of the function $S_{X,X'}$ on a process $\phi$ by wrapping it inside a new higher-order box, with dotted wires which represent the outside or environment  \[   \input{figs/lot_1.tikz} . \] 
When combs are applied to the environment they commute with the action of the $S_{X,X'}$ as follows \[   \input{figs/lot_1_outside.tikz} \ = \   \input{figs/lot_1_inside.tikz} . \]

We can also easily generalise this approach to give the multi-input categorical supermaps \cite{wilson_locality}.
A multi-input categorical supermap will naturally be a function on tuples of processes \[   S_{X_1 X_1' X_2 X_2' }: \cat{C}^{\mathbf{d}}(A_1 \otimes X_1 , A_1'   \otimes X_1'  ) \times \cat{C}^{\mathbf{d}}(A_2 \otimes X_2 , A_2'   \otimes X_2'  ) \rightarrow \cat{C}^{\mathbf{d}}(B \otimes X_1  \otimes X_2  , B'   \otimes X_1'  \otimes X_2'  ),  \] diagrammatically denoted as \[ \input{figs/lot_multi_basic.tikz},\]
and satisfying the following law ensuring commutation with combs in each input: \[   \input{figs/lot_multi_outside.tikz} \ = \ \input{figs/lot_multi_inside.tikz}.    \]

We now summarise this discussion by giving the definition of categorical supermaps.

\begin{definition}[Categorical Supermaps]
A one-input categorical supermap of type \[S: (A \Rightarrow A') \rightarrow (B \Rightarrow B')\] is a family of functions \[   S_{X,X'} :   \cat{C}^{\mathbf{d}}(A \otimes X, A' \otimes X') \rightarrow \cat{C}^{\mathbf{d}}(B \otimes X , B' \otimes X')  \] satisfying  \[   \input{figs/lot_1_outside.tikz} \ = \   \input{figs/lot_1_inside.tikz} . \]
Furthermore,  a multi-input categorical supermap of type \[    S:   \bigtimes_i (A_i \Rightarrow A_i') \rightarrow (B \Rightarrow B')  \] is a family of functions \[S_{X_1 X_1' X_2 X_2' }:  \bigtimes_i \cat{C}^{\mathbf{d}}(A_i \otimes X_i , A_i'   \otimes X_i'  ) \rightarrow \cat{C}^{\mathbf{d}}(B \otimes_i (X_i)  , B'   \otimes_i (X_i' ) )\] satisfying  \[   \input{figs/lot_multi_outside.tikz} \ = \ \input{figs/lot_multi_inside.tikz}.    \]
\end{definition}

In general, categorical supermaps can be reframed entirely in terms of the mathematical theory of strong profunctors \cite{hefford_supermaps, hefford2025bvcategoryspacetimeinterventions}. 
Doing so makes the definition significantly more compact and opens the door to the study of supermaps on general spaces, accounting for supermaps with indefinite order but also supermaps with definite order and even higher higher-order supermaps. 

We now complete this section by discussing the basic lemma on which all key concrete characterisation results in categorical supermaps are based.

\begin{lemma}\label{lem:important}
For any pair of parameterised instruments $M,M'$ with choices $x,x'$ and outcomes $y,y'$ such that: \[  \input{figs/param_inst3.tikz} \ = \   \input{figs/param_inst_4.tikz},     \]  we have the following equality \[  \input{figs/param_inst5.tikz} \ = \   \input{figs/param_inst6.tikz}.     \] 
\end{lemma}
\begin{proof}
Given in the appendix, a mild generalisation of the proof of Lemma 5 of \cite{wilson_locality}. 
\end{proof}

\subsection{The Yoneda Lemma for Compact Closed Categories}
We now prepare to state the second fundamental tool for categorical supermaps, which in fact hinges directly on the Yoneda lemma, one of the most important results in category theory more broadly \cite{maclane1998categories}.
A description of the Yoneda lemma and the reframing of the following results into categorical algebra is given in the Appendix.
Here, let us give a purely diagrammatic interpretation of the Yoneda lemma insofar as it applies to compact closed categories, since it is the key idea on top of which the Yoneda lemma for categorical supermaps will be built.
\begin{theorem}[Yoneda Lemma for Compact Closed Categories]\label{thm:yo_compact}
On any compact closed category $\cat{C}$, there is a natural isomorphism between categorical supermaps of type \[  S: (A \Rightarrow A') \rightarrow (B \Rightarrow B') \] and processes in $\cat{C}$ of type $s: B \otimes A' \rightarrow B'\otimes A$.
Furthermore, there is a natural isomorphism between categorical supermaps of type \[    S:   \bigtimes_i (A_i \Rightarrow A_i') \rightarrow (B \Rightarrow B')  \] and processes in $\cat{C}$ of type $s: B \otimes_{i}(A_i') \rightarrow B'\otimes_i(A_i)$.
\end{theorem}
\begin{proof}
This is the Yoneda Lemma; we will however give the purely graphical proof of \cite{wilson2023compositional}. In both cases, the proof works by constructing a potential representative map $S: A \otimes B' \rightarrow A' \otimes B$ by evaluating $S_{B'B}$ on the swap\footnote{Note that on a non-symmetric monoidal category with left and right duals, a similar argument can be performed but by evaluating on a term involving cups and identities, see the appendix for details. As a result, the Yoneda lemma for categorical supermaps also holds for non-symmetric monoidal categories.} \[  \input{figs/lot_swap.tikz}  . \] Now we check that the constructed $S$ acts as a CJ-supermap in the same way as the original categorical supermap $S_{XX'}$. Indeed, see that \begin{align}
    \input{figs/lot_2.tikz} \  = \ \input{figs/lot_3.tikz} \ = \ \input{figs/lot_1.tikz}.
\end{align}
Consequently, it is easy to see that from any categorical supermap one can construct a CJ-supermap, and from any CJ-supermap one can construct a categorical supermap, and that furthermore these constructions are both-ways invertible.
As a result, categorical supermaps are CJ-supermaps on categories with channel-state dualities.

The multi-input case follows similarly, for instance for two inputs, by the same reasoning, we have,
\begin{align}
    \input{figs/multi_yoneda_1.tikz} \  = \ \input{figs/multi_yoneda_2.tikz}  \\ 
    = \ \input{figs/multi_yoneda_3.tikz}\ =  \ \input{figs/lot_multi_basic.tikz}.
\end{align}
\end{proof}
We can conclude that on a generalised theory whose deterministic category $\cat{C}^{\mathbf{d}}$ is compact closed, any categorical supermap will find a representation as a CJ-supermap.
In practice, the deterministic category of a generalised theory never has channel-state duality, indeed if it does it becomes unsuitable for the study of causality \cite{coecke_causalcats}. 
This means we need more than the previous result in order to establish a representation theorem for categorical supermaps.

\section{A Yoneda Lemma for Categorical Supermaps}
In section \ref{sec:GTwDuality}, we presented a first way to define supermaps on generalised theories under the name of \textit{CJ-supermaps}. This is the way that has been effectively followed in most of the literature on higher-order processes thus far: assuming that the theory has a (non-deterministic) channel-state duality (what we call a \textit{generalised theory with channel-state duality}), one uses the duality to see the supermaps as a special kind of process in the theory. That way, the definition can be expressed purely as constraints on the processes of the theory, which greatly limits how much the formalism must be extended to accommodate the supermaps. For example, in quantum theory, a quantum superchannel can be defined in terms of constraints on the quantum channel that represents it via channel-state duality; this avoids the need to define a concept such as a completely-CPTP-preserving map between two CPTP maps and the pitfalls that may come with it. 

In section \ref{sec:CSM}, we presented a second way to define supermaps on generalised theories under the name of \textit{categorical supermaps}. This definition avoids the channel-state duality assumption, thereby widening its range of applicability. However, it comes at the price of the supermaps no longer being seen as processes. Rather, they are families of abstract functions from deterministic processes to deterministic processes. 

But for generalised theories with channel-state duality, this begs the question of whether the two definitions actually coincide. Can the categorical supermaps be represented as concrete supermaps, i.e. processes of the theory. In other words, are they CJ-supermaps?
Evidence in favour of a positive answer comes from the \textit{Yoneda lemma for compact closed categories}, which states that on compact closed categories, categorical supermaps coincide with CJ-supermaps. In this section, we will show that this result holds for any generalised theory with channel-state duality: any generalised theory has a notion of categorical supermaps, but when their non-deterministic processes $\cat{C}$ is a compact closed category (with cup and cap elements of instruments), these categorical supermaps have a concrete representation that exactly coincides with the notion of CJ-supermaps.


\YonedaLemma*

\begin{proof}
This follows by combining the techniques of lemma \ref{lem:important} and the Yoneda lemma.
Indeed, note that by assumption there exist instruments with components that are, up to scalars, the caps and cups.
We can use them to mimic the setup for the Yoneda lemma.
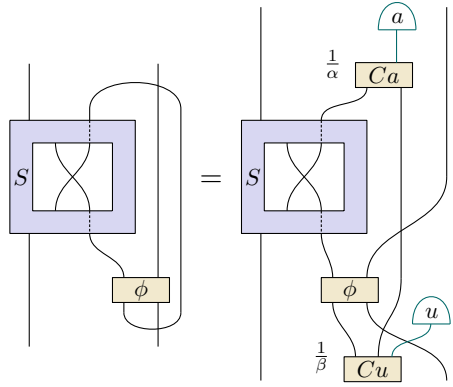
\[  \input{figs/fund_1.tikz} \ = \ \input{figs/fund_2.tikz}  \]
Using commutativity with deterministic combs, we can move the cap and cup instruments inside $S$ leaving only the classical post-selections outside, whereupon we can use lemma \ref{lem:important}.
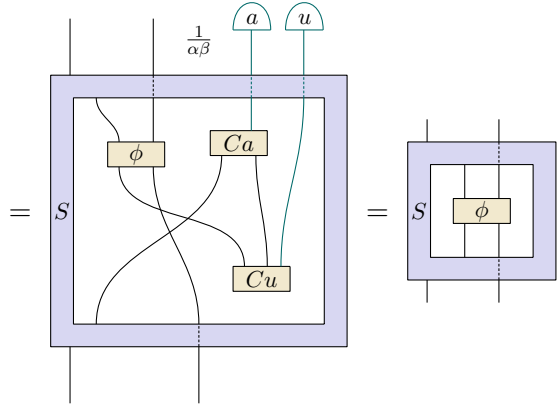
\[  = \ \input{figs/fund_3.tikz} \ = \ \input{figs/fund_4.tikz}  \]

The multi-input case follows similarly using the fact that we have commutativity separately in each input to $S$.
For instance, for the case of two inputs note that by fixing an input channel $\phi_1:A_1\otimes X_1\morph{}A_2\otimes X_2$ there is a family of functions constructed from $S$ of type,
\[  S_{X_1 X_1' X_2 X_2'}(\phi_1, -) : \cat{C}^{\mathbf{d}}(A_2 \otimes X_2 , A_2' \otimes X_2') \rightarrow \cat{C}^{\mathbf{d}}(B \otimes X_1 \otimes X_2 , B'\otimes X_1' \otimes X_2') . \]
This is a one-input categorical supermap, so we can use the earlier argument to yield,
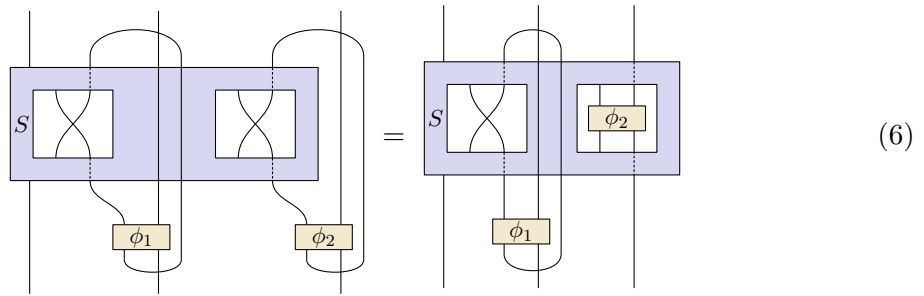
\begin{align}
    \input{figs/multi_yoneda_1.tikz} \  = \ \input{figs/multi_fund_1.tikz} 
\end{align}
and then by similar reasoning with $\phi_2$ fixed we have that
\begin{align}
    \input{figs/multi_fund_1.tikz}\ =  \ \input{figs/lot_multi_basic.tikz}.
\end{align}
This completes the proof. 
\end{proof}


As a consequence of our characterisation theorem we have the following results.

\Recovery*
\begin{proof}
Since Classical theory, Quantum theory, and Boxworld all have physical channel-state dualities it follows that the categorical supermaps on them are equal to the CJ-supermaps. 

\end{proof}

\subsection{Categorical supermaps on concrete profunctors}
As we mentioned, the adaptation of the Yoneda lemma we presented in this article was inspired by previous works relying on the theory of \textit{strong profunctors} \cite{hefford_supermaps, hefford2025bvcategoryspacetimeinterventions}. In this final section, we use strong profunctors to briefly present a generalisation of the domains for categorical supermaps for which the Yoneda lemma still holds. 

The definition of categorical supermaps used thus far relied on the idea that each `slot' of the supermap was meant to receive a process, and that the processes in each input were parallel to each other. But while the pair of systems used by each input process to interact with the supermap must effectively be in parallel, there is no reason why the pair of systems extending each of the processes should be as well. In other words, there is no reason that the environments accessible by each of the input processes should be assumed causally disconnected from one another. 

Take a bipartite supermap $\mathcal{S}$ for instance, the pair of input processes has the form $\phi_1 \otimes \phi_2$ where their types were assumed to be $\phi_1 : (A \otimes X) \rightarrow (A' \otimes X')$ and $\phi_2 : (B \otimes Y) \rightarrow (B' \otimes Y')$. In this case, the pair of systems used by $\phi_1$ lives in $A,A'$ and the one used by $\phi_2$ lives in $B,B'$. Now it makes sense that, if there were no extension (i.e. the `environment' wires $X,X',Y,Y'$ are all $I$), this pair of inputs has type $(A \rightarrow A') \otimes (B \rightarrow B')$. For example, in quantum theory, it makes sense that the inputs form a no-signalling channel together, or at the very least share past and future correlations. 
But one does not need such an assumption on the environments; for instance, one could have instead assumed that the environment on Alice's side is in the causal past of Bob's by requiring that $X' = Y$ and contracting Alice's output $X'$ with Bob's input $Y$, thereby restricting the space of legitimate supermaps do those which place Alice after Bob. To take into account these alternatives, that is, to account for supermaps on a broad class of types in the sense of \cite{kissinger_caus, Bisio_2019, hefford2025bvcategoryspacetimeinterventions, hoffreumon2022projective, jencova2024structurehigherorderquantum}, the notion of input processes can be generalised using \textit{profunctors}, and the definition of categorical supermaps can be generalised in accordance as we will now show. 
These profunctors where initially referred to as extension sets \cite{wilson_locality}, here we will refer to them in terms of their place in the broader theory of categorical supermaps, and instead refer to them as \textit{concrete} profunctors.



\begin{definition}[Concrete Profunctors]
A concrete profunctor $\cat{P}$ is a pair of objects $(A,A')$ along with for each pair of objects $X,X'$ a chosen subset \[  \cat{P}(X,X') \subseteq    \cat{C}^{\mathbf{d}}(A \otimes X, A' \otimes X'),  \] such that:
\begin{enumerate}
\item These sets are closed under action by combs on the environment, meaning that \[        \input{figs/conc_prof_1.tikz} \in \cat{P}(X,X') \ \implies  \ \input{figs/conc_prof_2.tikz} \in  \cat{P}(Y,Y'),   \]
for every $f$ and $g$.
\item The swap is included in $ \cat{P}(A',A)$, in other words \[    \input{figs/swap.tikz} \in \cat{P}(A',A) .  \]
\end{enumerate}
\end{definition}

Examples of concrete profunctors include:
\begin{itemize}
\item Spaces of product channels with shared past and future  \[  \cat{P}(X,X')  \ = \ \left\{  \input{figs/product_space_1.tikz} \right\}   .  \]
\item Spaces of channels with memory  \[    \cat{P}(X,X')   \ = \ \left\{ \input{figs/seq_space_1.tikz} \right\}  .   \]
\end{itemize}

Categorical supermaps can be defined on all strong profunctors, whereas it is only obvious how to generalise CJ-supermaps to concrete profunctors. As a result, it is on concrete profunctors where we might hope to be able to extend the Yoneda Lemma for categorical supermaps.
\begin{definition}[Categorical Supermaps on Concrete Profunctors]
A categorical supermap between concrete profunctors of type \[S: \cat{P} \rightarrow  \cat{Q}   \] is a family of functions \[   S_{X,X'} : \cat{P}(X,X') \rightarrow  \cat{Q}(X,X')  \] satisfying  \[   \input{figs/lot_1_outside.tikz} \ = \   \input{figs/lot_1_inside.tikz} . \]
Multi-input categorical supermaps on concrete profunctors are defined similarly.
\end{definition}
Let us also define the CJ-supermaps on concrete profunctors.
\begin{definition}
    A CJ-supermap of type \[ \cat{P} \rightarrow  \cat{Q},  \] with $\cat{P}$ associated to $A,A'$ and $\cat{Q}$ associated to $B,B'$ respectively, on a generalised theory with channel-state duality is a process of type $S: B \otimes A' \rightarrow B' \otimes A$ in $\cat{C}$ such that for every $\phi : A \otimes X \rightarrow A' \otimes X'$ in $\cat{P}(X,X')$ the following map is in $\cat{Q}(X,X')$. \[ \input{figs/cj_supermap_1.tikz}.  \]
\end{definition}

The Yoneda lemma for categorical supermaps extends to all concrete profunctors, meaning that categorical supermaps are equivalent to CJ-supermaps on a broad class of spaces. 
\begin{theorem}[Yoneda Lemma for Categorical Supermaps on Concrete Profunctors]
For any pair of concrete profunctors on a generalised theory, the categorical supermaps of type $\cat{P} \rightarrow \cat{Q} $ are equal to the CJ-supermaps of the same type.
\end{theorem}
\begin{proof}
The proof of the preliminary non-contextual assignment lemma and the induced Yoneda lemma for categorical supermaps are syntactically identical to the proofs already given in the simplified case. The only note to make, is that the proof of the Yoneda lemma for categorical supermaps requires the existence of a swap process, which is indeed given by definition for concrete profunctors.
\end{proof}
Beyond the scope of this paper, is the question of the exact type-theoretic properties of the composition operations introduced above. Their associativity, extension to higher-order, and interrelations such as currying, so that they might fit the general mould laid out for higher-order process theories \cite{Wilson_causal}, is left for future work.


\section{Conclusion}

The development of categorical supermaps is motivated by the straightforward aim of providing a stable compositional framework for higher-order processes in general physical theories.
More precisely, it aims to provide a framework for studying indefinite causal order in post-quantum theories and/or in infinite-dimensional theories. 
In this article, we brought the former goal of categorical supermaps to maturity via a Yoneda-like lemma for them. We have found that whenever a post-quantum theory has a CJ-isomorphism, the lemma supplies a concrete CJ-representation of what should be the adequate notion of supermap for this theory, thereby bypassing any guesswork. 
This approach also puts forward an alternative, weaker notion of higher-order map on Boxworld, motivated by the use of the non-signalling tensor product when framing Boxworld as a generalised theory. 

Natural future directions include: the development of a theory-independent formulation of the process-matrix axioms (as opposed to the theory-independent formulation of the supermap axioms); the identification of important subclasses of supermaps such as those which exhibit classical or generalised physical control of causal order \cite{PRXQuantum.2.030335}; a general study of the strength of indefinite causality (via the violation of causal inequalities \cite{Lugt:2023aa, oreshkov}) conceivable across the spectrum of possible physical theories; and a concrete characterisation of the cw-NSWSE Boxworld tensors and their achievable non-causal correlations. 



 \section*{Acknowledgements}

 	This work has been partially funded by the French National Research Agency (ANR) within the framework of ``Plan France 2030'', under the research projects EPIQ ANR-22-PETQ-0007 and HQI-R\&D ANR-22-PNCQ-0002. T.H. has been supported by the Slovak grants APVV-22-0570 and VEGA 2/0128/24, as well as the ID\# 62312 grant from the John Templeton Foundation, as part of the \href{https://www.templeton.org/grant/the-quantum-information-structure-of-spacetime-qiss-second-phase}{`The Quantum Information Structure of Spacetime' Project (QISS)}. The opinions expressed in this publication are those of the authors and do not necessarily reflect the views of the John Templeton Foundation. The authors are grateful to Jessica Bavaresco, Cyril Branciard, Kuntal Sengupta, Allesandro Bisio, Luca Apadula, Zixuan Liu, and Anna Jen\v{c}ova for useful discussions.

\bibliographystyle{plainnat}
\bibliography{bibliography.bib}

\appendix

\section{No signalling without system exchange and its generalisations}

The NS famillies of probabilistic operations are also be characterised in  \cite{bavaresco2024indefinitecausalorderboxworld} as those with marginals over $x'$ which decompose as either \[ \sum_{\lambda} \pi_{\lambda}  \,  \input{figs/ns_diagram_decomp.tikz} \qquad \text{or} \qquad \input{figs/ns_diagram_decomp2.tikz}   \] with $\pi_{\lambda}$ a probability distribution. We will refer to the former decomposition as the non-trivial decomposition for the marginal and the latter as the trivial decomposition for the marginal.

A Boxworld NSWSE tensor can be recharacterised as in \cite{bavaresco2024indefinitecausalorderboxworld} simply as an effect on the tensor product of combs \[W: (((C \rightarrow C') \rightarrow (D \rightarrow D')) \otimes ((A \rightarrow A') \rightarrow (B \rightarrow B'))) \rightarrow I \] depicted \[  \input{figs/unfold_tensor_1.tikz},  \]
such that \[   \input{figs/unfold_tensor_2.tikz} \ = \ \input{figs/unfold_tensor_3.tikz} \quad \text{and} \quad   \input{figs/unfold_tensor_2l.tikz} \ = \ \input{figs/unfold_tensor_3l.tikz}    .   \]  

\begin{theorem}
NSWSE Boxworlds tensors $\subseteq$ Boxworld Supermaps $\subseteq$ w-NSWSE Boxworld tensors
\end{theorem}
\begin{proof}
We show that if a matrix in $\mat{\reals^+}$ is a NSWSE Boxworld tensor, then indeed that is sufficient to conclude that it is a Boxworld supermap (equivalently a cw-NSWSE Boxworld tensor). 
Without loss of generality, we only need to check whether the following diagram is non-signalling: \[ \input{figs/boxworld_proof_single_1.tikz}, \] whenever the box $W$ satisfies the following right-handed NSWSE constraint  \[  \input{figs/rhs_nswse_1.tikz} \ = \ \input{figs/rhs_nswse_2.tikz}   \] for some $W'$.
Finally then, to confirm this we make use of the concrete characterisation theorem for the marginal over any NS process $T$ which gives for the contraction between $W$ and $T$ the following concrete expression in the non-trivial case (the proof of the trivial decomposition is immediate) \[ \sum_{\lambda} \pi_{\lambda}  \quad \input{figs/unfold_comb_1.tikz}  \quad = \quad  \sum_{\lambda} \pi_{\lambda}  \quad \input{figs/unfold_comb_2.tikz}    \] then using that each $D_{\lambda}$ is deterministic we have 
 \[ \sum_{\lambda} \pi_{\lambda}  \quad \input{figs/unfold_comb_3.tikz}  \quad = \quad  \sum_{\lambda} \pi_{\lambda}  \quad \input{figs/unfold_comb_4.tikz}    \]
We hence find the following 
 \[   \input{figs/boxworld_proof_single_2b.tikz} 
  \quad = \quad \input{figs/boxworld_proof_single_5.tikz}. \] The same reasoning process can be applied to each party, demonstrating that the induced state is multiparty non-signalling and so a state in Boxworld. This completes the proof of the first inclusion.
For the latter inclusion, note that since Boxworld supermaps are cw-NSWSE tensors, they are automativelly w-NSWSE tensors. 
\end{proof}

\section{Proof of Lemma \ref{lem:important}}
In this section we will prove that whenever a component of one parameterised instrument is equal to a component of another parameterised instrument, then the application of a categorical supermap will preserve this equality. 
Before we do so, we must confirm convex linearity of categorical supermaps. 

\begin{lemma}
Every component of a one-input categorical supermap on a generalised theory is a convex linear function.
\end{lemma}
\begin{proof}
We note that 
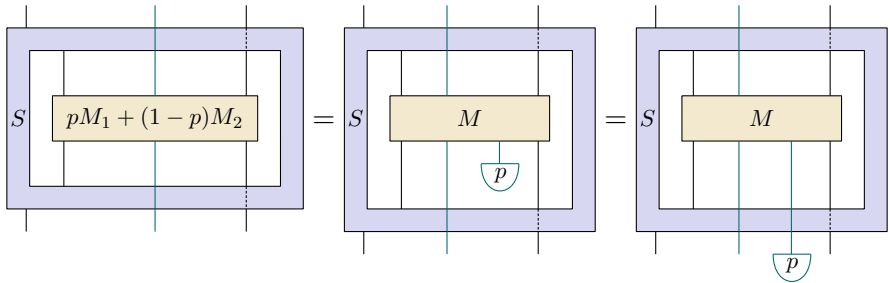
\[\input{figs/convex_1.tikz} = \input{figs/convex_2.tikz} = \input{figs/convex_3.tikz} \]
where $M$ is the controlled version of $M_1,M_2$ and $p :  I \rightarrow 2$ is the distribution $(p,(1-p))$. Consequently the above expands to
                     \[  = p \  \input{figs/convex_4.tikz} + (1-p) \  \input{figs/convex_5.tikz}, \] where using commutativity again gives
                     \[  =   p  \  \input{figs/convex_6.tikz} + (1-p) \  \input{figs/convex_7.tikz}, \] and re-)using the definition of ocntrol returns
                      \[ = p \   \input{figs/convex_8.tikz} + (1-p) \  \input{figs/convex_9.tikz}. \]
\end{proof}
Having proven convexity we are now well placed to prove that categorical supermaps preserve equality between components of parameterised instruments. 
\renewcommand{\thelemma}{\ref{lem:important}}
\begin{lemma}[Restated]\label{lem:important-restated}
For any pair of parameterised instruments $M,M'$ with choices $x,x'$ and outcomes $y,y'$ such that: \[  \input{figs/param_inst3.tikz} \ = \   \input{figs/param_inst_4.tikz},     \]  then \[  \input{figs/param_inst5.tikz} \ = \   \input{figs/param_inst6.tikz}.     \] 
\end{lemma}

\begin{proof}
The proof of this theorem follows directly the proof given in \cite{wilson_locality} for an almost identical lemma in which the wires are taken to be simply fully-quantum systems. The proof hinges only on a few basic properties of the classical interface.
Indeed, let us construct a pair of stochastic channels \[  \input{figs/param_sigma_1.tikz} \ := \   \input{figs/param_sigma_2.tikz} \ + \ \left(    \input{figs/param_sigma_3.tikz}   \right),  \quad \quad  \input{figs/param_sigma_4.tikz} \ := \   \input{figs/param_sigma_5.tikz} \ + \ \left(    \input{figs/param_sigma_6.tikz}   \right).   \]
It is easy to show (see \cite{wilson_locality} for the proof technique) that: \[  \input{figs/param_conv_1.tikz} \ - \ \input{figs/param_conv_2.tikz} \quad  = \quad  \input{figs/param_conv_3.tikz} \ - \ \input{figs/param_conv_4.tikz} . \] 
Now note that the purpose of introducing $\Sigma$ and $\Sigma'$ is to find a deterministic channel which encodes the difference between the outcomes $y$ and $y'$, more precisely we have the following key properties: \[  \input{figs/param_sigma_7.tikz} \ = \ \input{figs/param_sigma_8.tikz} \quad \quad  \input{figs/param_sigma_9.tikz} \ = \ \input{figs/param_sigma_10.tikz} . \]
From this we have that \[  \input{figs/param_inst5.tikz} \ - \ \input{figs/param_inst6.tikz} \quad  = \quad  \input{figs/param_inst_7.tikz} \ - \ \input{figs/param_inst_8.tikz} . \] 
By commutativity with deterministic channels we then have that this term is equal to \[  \input{figs/param_inst_9.tikz} \ - \ \input{figs/param_inst_10.tikz}, \] 
finally, using convex linearity one can hence derive the equality of this term\footnote{One might object that the linear combination here is not a convex combination, however it can be derived purely from convex linearity, see the proof of Lemma 9 of \cite{wilson2023compositional} for a full exposition of this point.} to \[  \input{figs/param_inst_11.tikz} \ - \ \input{figs/param_inst_12.tikz} \quad  = \quad  \input{figs/param_inst_13.tikz} \ - \ \input{figs/param_inst_14.tikz} \ = \ 0 .  \] This completes the proof. 
\end{proof}

\section{Strong Profunctors and the Yoneda Lemma}
As explained in \cite{hefford_supermaps} the definition of categorical supermaps can be stated completely within the language of \textit{strong profunctors} also known as Tambara modules \cite{tambara,pastro_street}.
The utility of making this identification is that the main results of this paper can be derived almost immediately by abstract nonsense as instances of the Yoneda lemma.
Let us now give a brief outline of how this works.
Readers interested in further details from the point of view of supermaps can find them in \cite{hefford_supermaps,hefford2025bvcategoryspacetimeinterventions}, from the point of view of monoidal context theory in \cite{roman_thesis,roman_optics} and from the point of view of profunctor optics as used in computer science in \cite{clarke_profunctor,riley_optics}.

Given a symmetric monoidal category $\cat{C}$, we have in many places in this article written expressions like $\cat{C}(A,B)$ and $\cat{C}(A\otimes X,A'\otimes X')$ for the sets of maps $A\morph{}B$ and $A\otimes X\morph{} A'\otimes X'$ respectively.
These objects are in fact instances profunctors: for example the family of sets $\{\cat{C}(A,B)\}_{A,B}$ is closed under composition on either side by morphisms from $\cat{C}$.
\begin{definition}[Profunctor]
    A profunctor $P:\cat{C}\pmorph\cat{C}$ is a functor $P:\opcat{\cat{C}}\times\cat{C}\morph{}\set$.
\end{definition}
In the case of $\cat{C}(-,-):\cat{C}\pmorph\cat{C}$, we send each pair $(A,B)$ of objects of $\cat{C}$ to the set $\cat{C}(A,B)$ of morphisms $A\morph{}B$ in $\cat{C}$.
Then for a morphism $g:B\morph{}B'$ of $\cat{C}$, we send it to the function which simply acts by post-composing with $g$,
\begin{equation*}
    f\in \cat{C}(A,B) \mapsto gf \in \cat{C}(A,B'),
\end{equation*}
and similarly by pre-composition for any morphism $h:A'\morph{}A$.
This special profunctor $\cat{C}(-,-)$ is known as the hom-profunctor and plays a rather special role in the theory.

The family of sets $\{\cat{C}(A\otimes X,A'\otimes X')\}_{X,X'}$ is also a profunctor $\cat{C}\pmorph\cat{C}$ in a similar way to the hom-profunctor, but just by composition along $X$ and $X'$.
Note that $A$ and $A'$ are fixed here, this is because these objects are the input to the supermap we will eventually define.
It is $X$ and $X'$ that vary, for these model the environment around $(A,A')$.
We write this profunctor as $\cat{C}(A\otimes-,A'\otimes-)$.

The profunctors we have so far are, in fact, not only compatible with composition in $\cat{C}$ but also with the tensor product of $\cat{C}$.
Such profunctors are known as \textit{strong}.

\begin{definition}[Strong Profunctor]
    A profunctor $P:\cat{C}\pmorph\cat{C}$ is strong if it equipped with a natural transformation $\zeta: \cat{C}(-,-)\times P(-,-) \morph{} P(-\otimes-,-\otimes-)$ such that a number of coherence conditions hold making $\zeta$ essentially associative and unital \cite{tambara,pastro_street}.
\end{definition}

A strength for a profunctor $P$ gives a way of ``tensoring'' elements of $P$ with morphisms of $\cat{C}$.
In the case of $\cat{C}(A\otimes-,A'\otimes-)$ the strength is simply given by the tensor product of $\cat{C}$.

Strong profunctors assemble into a category $\StProf(\cat{C})$ whose morphisms are the strong natural transformations -- these are simply natural transformations that commute with the strengths.
It turns out that this property captures precisely commutation with combs so that it is possible to restate the definition of a categorical supermap as follows.

\begin{definition}[Categorical Supermap]
    A categorical supermap is a strong natural transformation,
    \begin{equation*}
        S:\cat{C}(A\otimes -,A'\otimes -) \morph{} \cat{C}(B\otimes -,B'\otimes-).
    \end{equation*}
\end{definition}
This can be extended to the multi-input case, and we refer the reader to \cite{hefford_supermaps} for the details of this.

In the case that we are working with a CPT, categorical supermaps simply become natural transformations on the deterministic part of the theory,
\begin{equation*}
    S:\cat{C}^{\mathbf{d}}(A\otimes -,A'\otimes -) \morph{} \cat{C}^{\mathbf{d}}(B\otimes -,B'\otimes-),
\end{equation*}
where $\cat{C}^{\mathbf{d}}(A\otimes -,A'\otimes -):\cat{C}^{\mathbf{d}}\pmorph\cat{C}^{\mathbf{d}}$ so that we only require $S$ to commute with deterministic combs.

One of the main results about strong profunctors is their representation theory developed in \cite{pastro_street}.
This result characterises strong profunctors as representations (copresheaves) of a certain category $\opt(\cat{C})$ of \textit{coend optics}.
The reason for this name comes from their application to bidirectional data accessors \cite{clarke_profunctor} but the more quantum-inclined readers of this paper might be intrigued to learn that it bears close resemblance to combs.

\begin{definition}[Coend Optics]
    $\opt(\cat{C})$ is the category whose objects $(A,A')$ are pairs of objects from $\cat{C}$ and whose morphisms $(A,A')\morph{}(B,B')$ are elements of the following set,
    \begin{equation*}
        \int^X \cat{C}(B,X\otimes A)\times \cat{C}(X\otimes A',B').
    \end{equation*}
\end{definition}
The integral symbol here is not actually an integral at all and is known as a \textit{coend} \cite{loregian_coend}.
The details of coends are not terribly important for this informal discussion, and one may effectively think of the earlier expression as an abstract way of describing the space of one-holed combs with input $(A,A')$ and output $(B,B')$.

\begin{remark}
    One may wonder why we do not just call this category $\comb(\cat{C})$ and this is because there is subtlety in the difference between coend optics and combs as envisaged in quantum theory.
    The equivalence relation on the sets of combs and the set of coend optics depends on the underlying category $\cat{C}$ and is not always the same.
    See \cite{hefford_optics,hefford_supermaps} for further discussion of this, but note that in the case that $\cat{C}$ is compact closed, coend optics and combs coincide so there is no danger in making this correspondence.
\end{remark}

We are now in a position to state the representation theorem.
\begin{theorem}[\cite{pastro_street}]
    $\StProf(\cat{C}) \cong [\opt(\cat{C}),\set]$.
\end{theorem}

This is \textit{the} theorem that allows us to reduce the proof of Theorem \ref{thm:yo_compact} to a couple of lines and justify its naming as the Yoneda lemma for compact closed categories.

\begin{proof}
    Note that 
    \begin{align*} 
        y_{(A',A)} & = \int^X \cat{C}(-,X\otimes A')\times \cat{C}(X\otimes A,-) \cong \int^X \cat{C}(-,X\otimes A')\times \cat{C}(X,A^*\otimes -) \\
        & \cong \cat{C}(-,A^*\otimes A'\otimes-) \cong \cat{C}(A\otimes-,A'\otimes-)
    \end{align*}
    so that $\suphom{A}{A'}$ is representable.
    The result now follows immediately by the Yoneda lemma.
    \begin{align*}
        \mathrm{Nat}\big(\suphom{A}{A'},\suphom{B}{B'}\big) & \cong \mathrm{Nat}\big(y_{(A',A)}, \suphom{B}{B'}\big) \\
        & \cong \cat{C}(B\otimes A',B'\otimes A)
    \end{align*}
\end{proof}
The element in $\cat{C}(B\otimes A',B'\otimes A)$ that is equivalent to the original supermap is given in the original proof by the action of $S$ on the swap.
The Yoneda lemma is precisely the reason that this appears.

\end{document}

%% file: preamble.tex
\usepackage{amsmath}
\usepackage{amsthm}
\usepackage{amssymb}
\usepackage{physics}
\usepackage{stmaryrd}
\usepackage{tikz}
\usepackage{tikz-cd}
\usepackage{tikzit}
\usepackage{circuitikz}
\usepackage[numbers,sort&compress]{natbib}
\usepackage{hyperref}
\usepackage{tubes}
\usepackage{xsavebox}
\usepackage{cmll}

\input{styles.tikzdefs}
\input{styles.tikzstyles}
\usetikzlibrary{circuits.ee.IEC}

\newcommand{\morph}[1]{\xrightarrow{#1}}

\newcommand{\pmorph}{\relbar\joinrel\mapstochar\joinrel\rightarrow}
\newcommand{\cat}[1]{\mathcal{#1}}
\newcommand{\cpm}[1]{\mathsf{CPM}(#1)}

\newcommand{\mat}[1]{\mathsf{Mat}_{#1}}
\newcommand{\opcat}[1]{#1^\textrm{op}}
\newcommand{\set}{\mathsf{Set}}

\newcommand{\nats}{\mathbb{N}}
\newcommand{\reals}{\mathbb{R}}

\newcommand{\StProf}{\mathsf{StProf}}

\newcommand{\opt}{\mathsf{Optic}}
\newcommand{\comb}{\mathsf{Comb}}

\newcommand{\stoch}{\mathsf{Stoch}}

\newcommand{\bl}{\mathord{=}}
\newcommand{\suphom}[2]{\cat{C}(#1 \otimes-, #2 \otimes\bl)}

\newcommand{\seq}{\olessthan}

\makeatletter
\newcommand{\xRightarrow}[2][]{\ext@arrow 0359\Rightarrowfill@{#1}{#2}}
\makeatother

%% file: styles.tikzstyles
\tikzstyle{discarding}=[fill=white, draw=black, shape=circle, style=upground]
\tikzstyle{smalldiscarding}=[fill=white, draw=black, style=upground, scale=0.75]
\tikzstyle{backdiscard}=[fill=white, draw=black, shape=circle, style=downground, scale=0.5]
\tikzstyle{smallbackdiscard}=[fill=white, draw=black, shape=circle, style=downground, scale=0.5]
\tikzstyle{state}=[fill=white, draw=black, style=triang, tikzit shape=rectangle]
\tikzstyle{kstate}=[fill=white, draw=black, style=kpoint, tikzit shape=rectangle]
\tikzstyle{kstateconj}=[fill=white, draw=black, style=kpoint conjugate, tikzit shape=rectangle]
\tikzstyle{kstateBIG}=[fill=white, draw=black, style=big kpoint, tikzit shape=rectangle]
\tikzstyle{effect}=[fill=white, draw=black, style=triangdag]
\tikzstyle{keffect}=[fill=white, draw=black, style=kpoint adjoint]
\tikzstyle{keffectconj}=[fill=white, draw=black, style=kpoint transpose]
\tikzstyle{morphdag}=[style=mapdag]
\tikzstyle{morph}=[style=hadamard]
\tikzstyle{WIDEmorph}=[style=hadamard, minimum width=14mm]
\tikzstyle{morphtrans}=[style=maptrans]
\tikzstyle{morphconj}=[style=mapconj]
\tikzstyle{CPMmorph}=[style=dmap]
\tikzstyle{CPMmorphconj}=[style=dmapconj]
\tikzstyle{CPMmorphdag}=[style=dmapdag]
\tikzstyle{CPMmorphtrans}=[style=dmaptrans]
\tikzstyle{CPMstate}=[fill=white, draw=black, style=triang, doubled]
\tikzstyle{CPMstateBIG}=[fill=white, draw=black, style={triang_lesssep}, doubled]
\tikzstyle{CPMkstate}=[fill=white, draw=black, style=kpoint, tikzit shape=rectangle, doubled]
\tikzstyle{CPMkstateconj}=[fill=white, draw=black, style=kpoint conjugate, tikzit shape=rectangle, doubled]
\tikzstyle{CPMkstateBIG}=[fill=white, draw=black, style=big kpoint, tikzit shape=rectangle, doubled]
\tikzstyle{CPMkeffect}=[fill=white, draw=black, style=kpoint adjoint, doubled]
\tikzstyle{CPMkeffectconj}=[fill=white, draw=black, style=kpoint transpose, doubled]
\tikzstyle{UHfB}=[fill=white, draw=black, style=triangdag, doubled, inner sep=-2pt]
\tikzstyle{leak}=[style=tinypoint, regular polygon rotate=-90]
\tikzstyle{leakfill}=[style=tinypoint, regular polygon rotate=-90, fill=black]
\tikzstyle{Z}=[style=dot, fill=green]
\tikzstyle{X}=[style=dot, fill=red]
\tikzstyle{black_dot}=[style=dot, fill=black]
\tikzstyle{white_dot}=[style=dot, fill=white]
\tikzstyle{qblack_dot}=[style=ddot, fill=black]
\tikzstyle{qwhite_dot}=[style=ddot, fill=white]
\tikzstyle{whitephase}=[style=wphase dot, fill=white]
\tikzstyle{qredphase}=[style=phase dot, fill=red]
\tikzstyle{qgreenphase}=[style=phase dot, fill=green]
\tikzstyle{had}=[style=hadamard, doubled]
\tikzstyle{box}=[style=hadamard]
\tikzstyle{bigbox}=[style=hadamard, minimum height=4mm, minimum width=8mm]
\tikzstyle{classhad}=[style=hadamard]
\tikzstyle{antipode}=[style=anti]

\tikzstyle{dottededge}=[-, dash pattern=on 1pt off 0.7pt]
\tikzstyle{double edge}=[-, style=doubled, draw=black, tikzit draw={rgb,255: red,18; green,168; blue,191}]
\tikzstyle{arrow}=[->]
\tikzstyle{new edge style 1}=[-, draw={rgb,255: red,242; green,233; blue,206}, fill={rgb,255: red,242; green,233; blue,206}]
\tikzstyle{morphism_shade}=[-, draw=black, fill={rgb,255: red,242; green,233; blue,206}, line join=bevel]
\tikzstyle{supermap_shade}=[-, fill={rgb,255: red,216; green,215; blue,242}, draw=black, line join=bevel]
\tikzstyle{hole_shade}=[-, fill=white, draw=black, line join=bevel]
\tikzstyle{new edge style 2}=[-, draw={rgb,255: red,14; green,188; blue,83}]
\tikzstyle{new edge style 3}=[<-, draw={rgb,255: red,234; green,209; blue,255}]
\tikzstyle{new edge style 4}=[<-, draw={rgb,255: red,255; green,0; blue,0}]
\tikzstyle{new edge style 5}=[-, draw={rgb,255: red,214; green,110; blue,62}]
\tikzstyle{new edge style 6}=[-, draw={rgb,255: red,174; green,20; blue,174}]
\tikzstyle{new edge style 0}=[-, fill=none, draw={rgb,255: red,255; green,0; blue,0}]
\tikzstyle{green_dotted}=[-, draw={rgb,255: red,0; green,106; blue,106}, dash pattern=on 1pt off 0.7pt]

%% file: figs/snake_2.tikz
\begin{tikzpicture}[scale=0.6, {every node/.style}={scale=0.8}]
	\begin{pgfonlayer}{nodelayer}
		\node [style=none] (2) at (-0.75, -1) {};
		\node [style=none] (3) at (0.75, 1) {};
	\end{pgfonlayer}
	\begin{pgfonlayer}{edgelayer}
		\draw [in=270, out=90] (2.center) to (3.center);
	\end{pgfonlayer}
\end{tikzpicture}

%% file: main_supermaps.bbl
\begin{thebibliography}{72}
\providecommand{\natexlab}[1]{#1}
\providecommand{\url}[1]{\texttt{#1}}
\expandafter\ifx\csname urlstyle\endcsname\relax
  \providecommand{\doi}[1]{doi: #1}\else
  \providecommand{\doi}{doi: \begingroup \urlstyle{rm}\Url}\fi

\bibitem[Abramsky and Coecke(2004)]{abramsky_coecke}
Samson Abramsky and Bob Coecke.
\newblock A categorical semantics of quantum protocols.
\newblock In \emph{Proceedings of the 19th Annual IEEE Symposium on Logic in
  Computer Science, 2004}, pages 415--425, 2004.
\newblock \doi{10.1109/LICS.2004.1319636}.

\bibitem[Apadula et~al.(2024)Apadula, Bisio, and
  Perinotti]{apadula2022nosignalling}
Luca Apadula, Alessandro Bisio, and Paolo Perinotti.
\newblock No-signalling constrains quantum computation with indefinite causal
  structure.
\newblock \emph{{Quantum}}, 8:\penalty0 1241, February 2024.
\newblock ISSN 2521-327X.
\newblock \doi{10.22331/q-2024-02-05-1241}.
\newblock URL \url{https://doi.org/10.22331/q-2024-02-05-1241}.

\bibitem[Barrett(2007)]{barrett_gpts}
Jonathan Barrett.
\newblock Information processing in generalized probabilistic theories.
\newblock \emph{Phys. Rev. A}, 75:\penalty0 032304, Mar 2007.
\newblock \doi{10.1103/PhysRevA.75.032304}.

\bibitem[Baumeler and Wolf(2016)]{baumeler_logically}
{\"A}min Baumeler and Stefan Wolf.
\newblock The space of logically consistent classical processes without causal
  order.
\newblock \emph{New Journal of Physics}, 18\penalty0 (1):\penalty0 013036,
  2016.
\newblock \doi{10.1088/1367-2630/18/1/013036}.

\bibitem[Baumeler et~al.(2014)Baumeler, Feix, and
  Wolf]{baumeler_incompatibility}
\"Amin Baumeler, Adrien Feix, and Stefan Wolf.
\newblock Maximal incompatibility of locally classical behavior and global
  causal order in multiparty scenarios.
\newblock \emph{Phys. Rev. A}, 90:\penalty0 042106, 2014.
\newblock \doi{10.1103/PhysRevA.90.042106}.

\bibitem[Bavaresco et~al.(2024)Bavaresco, Baumeler, Guryanova, and
  Budroni]{bavaresco2024indefinitecausalorderboxworld}
Jessica Bavaresco, {\"A}min Baumeler, Yelena Guryanova, and Costantino Budroni.
\newblock Indefinite causal order in boxworld theories, 2024.
\newblock URL \url{https://arxiv.org/abs/2411.00951}.

\bibitem[Bisio and Perinotti(2019)]{Bisio_2019}
Alessandro Bisio and Paolo Perinotti.
\newblock Theoretical framework for higher-order quantum theory.
\newblock \emph{Proceedings of the Royal Society A: Mathematical, Physical and
  Engineering Sciences}, 475\penalty0 (2225):\penalty0 20180706, 2019.
\newblock ISSN 1471-2946.
\newblock \doi{10.1098/rspa.2018.0706}.
\newblock URL \url{http://dx.doi.org/10.1098/rspa.2018.0706}.

\bibitem[Boisseau et~al.(2022)Boisseau, Nester, and
  Rom{\'a}n]{boisseau_cornering}
Guillaume Boisseau, Chad Nester, and Mario Rom{\'a}n.
\newblock Cornering optics, 2022.
\newblock URL \url{https://doi.org/10.48550/ARXIV.2205.00842}.

\bibitem[Chiribella and Liu(2022)]{Chiribella_Liu}
Giulio Chiribella and Zixuan Liu.
\newblock Quantum operations with indefinite time direction.
\newblock \emph{Communications Physics}, 5\penalty0 (1), 7 2022.
\newblock ISSN 2399-3650.
\newblock \doi{10.1038/s42005-022-00967-3}.
\newblock URL \url{http://dx.doi.org/10.1038/s42005-022-00967-3}.

\bibitem[Chiribella et~al.(2008{\natexlab{a}})Chiribella, D'Ariano, and
  Perinotti]{chiribella_circuits}
Giulio Chiribella, Giacomo~Mauro D'Ariano, and Paolo Perinotti.
\newblock Quantum circuit architecture.
\newblock \emph{Physical Review Letters}, 101\penalty0 (6):\penalty0 060401,
  2008{\natexlab{a}}.
\newblock \doi{10.1103/physrevlett.101.060401}.

\bibitem[Chiribella et~al.(2008{\natexlab{b}})Chiribella, D'Ariano, and
  Perinotti]{chiribella_supermaps}
Giulio Chiribella, Giacomo~Mauro D'Ariano, and Paolo Perinotti.
\newblock Transforming quantum operations: Quantum supermaps.
\newblock \emph{{EPL} (Europhysics Letters)}, 83\penalty0 (3):\penalty0 30004,
  2008{\natexlab{b}}.
\newblock \doi{10.1209/0295-5075/83/30004}.

\bibitem[Chiribella et~al.(2009)Chiribella, D'Ariano, and
  Perinotti]{chiribella_networks}
Giulio Chiribella, Giacomo~Mauro D'Ariano, and Paolo Perinotti.
\newblock Theoretical framework for quantum networks.
\newblock \emph{Phys. Rev. A}, 80:\penalty0 022339, 2009.
\newblock \doi{10.1103/PhysRevA.80.022339}.

\bibitem[Chiribella et~al.(2010)Chiribella, D'Ariano, and
  Perinotti]{chiribella_purification}
Giulio Chiribella, Giacomo~Mauro D'Ariano, and Paolo Perinotti.
\newblock Probabilistic theories with purification.
\newblock \emph{Phys. Rev. A}, 81:\penalty0 062348, 2010.
\newblock \doi{10.1103/PhysRevA.81.062348}.

\bibitem[Chiribella et~al.(2011)Chiribella, D'Ariano, and
  Perinotti]{chiribella_informational}
Giulio Chiribella, Giacomo~Mauro D'Ariano, and Paolo Perinotti.
\newblock Informational derivation of quantum theory.
\newblock \emph{Phys. Rev. A}, 84:\penalty0 012311, 2011.
\newblock \doi{10.1103/PhysRevA.84.012311}.

\bibitem[Chiribella et~al.(2013{\natexlab{a}})Chiribella, D'Ariano, Perinotti,
  and Valiron]{chiribella_switch}
Giulio Chiribella, Giacomo~Mauro D'Ariano, Paolo Perinotti, and Beno\^it
  Valiron.
\newblock Quantum computations without definite causal structure.
\newblock \emph{Phys. Rev. A}, 88:\penalty0 022318, 2013{\natexlab{a}}.
\newblock \doi{10.1103/PhysRevA.88.022318}.

\bibitem[Chiribella et~al.(2013{\natexlab{b}})Chiribella, Toigo, and
  Umanit{\`a}]{Chiribella2013NormalCP}
Giulio Chiribella, Alessandro Toigo, and Veronica Umanit{\`a}.
\newblock Normal completely positive maps on the space of quantum operations.
\newblock \emph{Open Systems \& Information Dynamics}, 20\penalty0
  (01):\penalty0 1350003, 2013{\natexlab{b}}.
\newblock \doi{10.1142/S1230161213500030}.
\newblock URL \url{https://doi.org/10.1142/S1230161213500030}.

\bibitem[Choi(1975)]{Choi_1975}
Man-Duen Choi.
\newblock Positive linear maps on complex matrices.
\newblock \emph{Linear Algebra and its Applications}, 10\penalty0 (3):\penalty0
  285--290, 1975.
\newblock ISSN 00243795.
\newblock \doi{10.1016/0024-3795(75)90075-0}.
\newblock URL \url{http://dx.doi.org/10.1016/0024-3795(75)90075-0}.

\bibitem[Clarke et~al.(2024)Clarke, Elkins, Gibbons, Loregian, Milewski,
  Pillmore, and Rom{\'{a}}n]{clarke_profunctor}
Bryce Clarke, Derek Elkins, Jeremy Gibbons, Fosco Loregian, Bartosz Milewski,
  Emily Pillmore, and Mario Rom{\'{a}}n.
\newblock Profunctor {O}ptics, a {C}ategorical {U}pdate.
\newblock \emph{{Compositionality}}, 6:\penalty0 1, February 2024.
\newblock ISSN 2631-4444.
\newblock \doi{10.32408/compositionality-6-1}.

\bibitem[Coecke(2010)]{coecke_picturalism}
Bob Coecke.
\newblock Quantum picturalism.
\newblock \emph{Contemporary Physics}, 51\penalty0 (1):\penalty0 59--83, 2010.
\newblock \doi{10.1080/00107510903257624}.

\bibitem[Coecke and Lal(2013)]{coecke_causalcats}
Bob Coecke and Raymond Lal.
\newblock Causal categories: Relativistically interacting processes.
\newblock \emph{Found Phys}, 43:\penalty0 458--501, 2013.
\newblock \doi{10.1007/s10701-012-9646-8}.

\bibitem[Coecke et~al.(2018)Coecke, Selby, and Tull]{coecke_classicality}
Bob Coecke, John Selby, and Sean Tull.
\newblock Two roads to classicality.
\newblock \emph{EPTCS}, 266:\penalty0 104--118, 2018.
\newblock \doi{10.4204/eptcs.266.7}.

\bibitem[Deutsch(1989)]{deutsch1989quantum}
David Deutsch.
\newblock Quantum computational networks.
\newblock \emph{Proceedings of the Royal Society of London. A. Mathematical and
  Physical Sciences}, 425\penalty0 (1868):\penalty0 73--90, 1989.
\newblock \doi{10.1098/rspa.1989.0099}.

\bibitem[Earnshaw et~al.(2023)Earnshaw, Hefford, and Rom{\'a}n]{earnshaw}
Matt Earnshaw, James Hefford, and Mario Rom{\'a}n.
\newblock The produoidal algebra of process decomposition, 2023.
\newblock URL \url{https://doi.org/10.48550/ARXIV.2301.11867}.

\bibitem[Giacomini et~al.(2016)Giacomini, Castro-Ruiz, and
  Brukner]{Giacomini_2016}
Flaminia Giacomini, Esteban Castro-Ruiz, and {\v C}aslav Brukner.
\newblock Indefinite causal structures for continuous-variable systems.
\newblock \emph{New Journal of Physics}, 18\penalty0 (11):\penalty0 113026,
  2016.
\newblock \doi{10.1088/1367-2630/18/11/113026}.
\newblock URL \url{https://dx.doi.org/10.1088/1367-2630/18/11/113026}.

\bibitem[Gogioso(2017)]{gogioso_fantastic}
Stefano Gogioso.
\newblock Fantastic quantum theories and where to find them.
\newblock \emph{arXiv Preprint}, arXiv:1703.10576, 2017.

\bibitem[Gogioso(2019)]{gogioso_cpm}
Stefano Gogioso.
\newblock Higher-order cpm constructions.
\newblock \emph{EPTCS}, 287:\penalty0 145--162, 2019.
\newblock \doi{10.4204/EPTCS.287.8}.

\bibitem[Gogioso and Scandolo(2018)]{gogioso_cpt}
Stefano Gogioso and Carlo~Maria Scandolo.
\newblock Categorical probabilistic theories.
\newblock \emph{EPTCS}, 266:\penalty0 367--385, 2018.
\newblock \doi{10.4204/EPTCS.266.23}.

\bibitem[Gogioso and Scandolo(2019)]{gogioso_hyper}
Stefano Gogioso and Carlo~Maria Scandolo.
\newblock Density hypercubes, higher order interference and hyper-decoherence:
  A categorical approach.
\newblock In \emph{Quantum Interaction}, pages 141--160. Springer International
  Publishing, 2019.
\newblock \doi{10.1007/978-3-030-35895-2_10}.

\bibitem[Guglielmi(2007)]{guglielmi}
Alessio Guglielmi.
\newblock A system of interaction and structure.
\newblock \emph{ACM Transactions on Computational Logic}, 8\penalty0 (1), 2007.
\newblock \doi{10.1145/1182613.1182614}.

\bibitem[Gutoski and Watrous(2007)]{Gutoski_2007}
Gus Gutoski and John Watrous.
\newblock Toward a general theory of quantum games.
\newblock In \emph{Proceedings of the Thirty-Ninth Annual ACM Symposium on
  Theory of Computing}, STOC '07, pages 565--574, New York, NY, USA, 2007.
  Association for Computing Machinery.
\newblock ISBN 9781595936318.
\newblock \doi{10.1145/1250790.1250873}.
\newblock URL \url{https://doi.org/10.1145/1250790.1250873}.

\bibitem[Hefford and Gogioso(2021{\natexlab{a}})]{hefford_galois}
James Hefford and Stefano Gogioso.
\newblock {CPM} categories for galois extensions.
\newblock In \emph{Proceedings QPL 2021}, volume 343, pages 165--192. Open
  Publishing Association, 2021{\natexlab{a}}.
\newblock \doi{10.4204/eptcs.343.9}.

\bibitem[Hefford and Gogioso(2021{\natexlab{b}})]{hefford_hypercubes}
James Hefford and Stefano Gogioso.
\newblock Hyper-decoherence in density hypercubes.
\newblock In \emph{Proceedings QPL 2020}, volume 340, pages 141--159. Open
  Publishing Association, 2021{\natexlab{b}}.
\newblock \doi{10.4204/eptcs.340.7}.

\bibitem[Hefford and Rom{\'a}n(2023)]{hefford_optics}
James Hefford and Mario Rom{\'a}n.
\newblock Optics for premonoidal categories.
\newblock In \emph{Proceedings ACT 2023}, volume 397, pages 152--171. Open
  Publishing Association, 2023.
\newblock \doi{10.4204/eptcs.397.10}.
\newblock URL \url{http://dx.doi.org/10.4204/EPTCS.397.10}.

\bibitem[Hefford and Wilson(2024)]{hefford_supermaps}
James Hefford and Matt Wilson.
\newblock {A Profunctorial Semantics for Quantum Supermaps}.
\newblock In \emph{Proceedings of the 39th Annual ACM/IEEE Symposium on Logic
  in Computer Science}, LICS '24, New York, NY, USA, 2024. Association for
  Computing Machinery.
\newblock \doi{10.1145/3661814.3662123}.

\bibitem[Hefford and Wilson(2025)]{hefford2025bvcategoryspacetimeinterventions}
James Hefford and Matt Wilson.
\newblock A bv-category of spacetime interventions, 2025.
\newblock URL \url{https://arxiv.org/abs/2502.19022}.

\bibitem[Hefford and Wilson(2026)]{hefford_decohere_causal}
James Hefford and Matt Wilson.
\newblock Decoherence to quantum theory from a causally indefinite post-quantum
  theory.
\newblock \emph{Phys. Rev. A}, 113:\penalty0 042433, Apr 2026.
\newblock \doi{10.1103/kmmy-3dy3}.
\newblock URL \url{https://link.aps.org/doi/10.1103/kmmy-3dy3}.

\bibitem[Hoffreumon and Oreshkov(2022)]{hoffreumon2022projective}
Timoth{\'e}e Hoffreumon and Ognyan Oreshkov.
\newblock Projective characterization of higher-order quantum transformations,
  2022.
\newblock URL \url{https://doi.org/10.48550/arXiv.2206.06206}.

\bibitem[Jamio{\l}kowski(1972)]{Jamiolkowski_1972}
Andrej Jamio{\l}kowski.
\newblock {Linear transformations which preserve trace and positive
  semidefiniteness of operators}.
\newblock \emph{Reports on Mathematical Physics}, 3\penalty0 (4):\penalty0
  275--278, 12 1972.
\newblock ISSN 00344877.
\newblock \doi{10.1016/0034-4877(72)90011-0}.
\newblock URL
  \url{http://linkinghub.elsevier.com/retrieve/pii/0034487772900110}.

\bibitem[Jen{\v c}ov{\'a}(2024)]{jencova2024structurehigherorderquantum}
Anna Jen{\v c}ov{\'a}.
\newblock On the structure of higher order quantum maps.
\newblock 2024.
\newblock URL \url{https://arxiv.org/abs/2411.09256}.

\bibitem[Joyal et~al.(1996)Joyal, Street, and Verity]{Joyal_Street_Verity_1996}
Andr{\'e} Joyal, Ross Street, and Dominic Verity.
\newblock Traced monoidal categories.
\newblock \emph{Mathematical Proceedings of the Cambridge Philosophical
  Society}, 119\penalty0 (3):\penalty0 447--468, 1996.
\newblock \doi{10.1017/S0305004100074338}.

\bibitem[Kissinger and Uijlen(2019)]{kissinger_caus}
Aleks Kissinger and Sander Uijlen.
\newblock {A categorical semantics for causal structure}.
\newblock \emph{Logical Methods in Computer Science}, 15, 2019.
\newblock ISSN 18605974.
\newblock \doi{10.23638/LMCS-15(3:15)2019}.

\bibitem[Lane(1998)]{maclane1998categories}
Saunders~Mac Lane.
\newblock \emph{Categories for the Working Mathematician}.
\newblock Graduate Texts in Mathematics. Springer Science+Business Media, New
  York, NY, 2 edition, September 1998.
\newblock ISBN 978-0-387-98403-2, 978-1-4419-3123-8, 978-1-4757-4721-8.
\newblock \doi{10.1007/978-1-4757-4721-8}.
\newblock URL \url{https://doi.org/10.1007/978-1-4757-4721-8}.
\newblock eBook package: Springer Book Archive; Copyright Springer
  Science+Business Media New York 1978.

\bibitem[Liu and Oreshkov(2025)]{liu2025parityerasurefoundationalprinciple}
Zixuan Liu and Ognyan Oreshkov.
\newblock Parity erasure: a foundational principle for indefinite causal order,
  2025.
\newblock URL \url{https://arxiv.org/abs/2512.08635}.

\bibitem[Loregian(2021)]{loregian_coend}
Fosco Loregian.
\newblock \emph{(Co)end Calculus}.
\newblock London Mathematical Society Lecture Note Series. Cambridge University
  Press, 2021.
\newblock \doi{10.1017/9781108778657}.

\bibitem[Mrini and Hardy(2024)]{mrini2024indefinitecausalstructurecausal}
Luke Mrini and Lucien Hardy.
\newblock Indefinite causal structure and causal inequalities with
  time-symmetry, 2024.
\newblock URL \url{https://arxiv.org/abs/2406.18489}.

\bibitem[Oreshkov et~al.(2012)Oreshkov, Costa, and Brukner]{oreshkov}
Ognyan Oreshkov, Fabio Costa, and {\v C}aslav Brukner.
\newblock Quantum correlations with no causal order.
\newblock \emph{Nature Communications}, 3\penalty0 (1092), 2012.
\newblock \doi{10.1038/ncomms2076}.

\bibitem[Pastro and Street(2008)]{pastro_street}
Craig Pastro and Ross Street.
\newblock Doubles for monoidal categories.
\newblock \emph{Theory and Applications of Categories}, 21\penalty0
  (4):\penalty0 61--75, 2008.

\bibitem[Paunkovic and Vojinovic(2023)]{paunkovic2023challenges}
Nikola Paunkovic and Marko Vojinovic.
\newblock Challenges for extensions of the process matrix formalism to quantum
  field theory, 2023.
\newblock URL \url{https://arxiv.org/abs/2310.04597}.

\bibitem[Pinzani et~al.(2019)Pinzani, Gogioso, and
  Coecke]{pinzani2019categoricalsemanticstimetravel}
Nicola Pinzani, Stefano Gogioso, and Bob Coecke.
\newblock Categorical semantics for time travel, 2019.
\newblock URL \url{https://arxiv.org/abs/1902.00032}.

\bibitem[Pollock et~al.(2018)Pollock, Rodr\'{\i}guez-Rosario, Frauenheim,
  Paternostro, and Modi]{process_tensor}
Felix~A. Pollock, C\'esar Rodr\'{\i}guez-Rosario, Thomas Frauenheim, Mauro
  Paternostro, and Kavan Modi.
\newblock Non-markovian quantum processes: Complete framework and efficient
  characterization.
\newblock \emph{Phys. Rev. A}, 97:\penalty0 012127, Jan 2018.
\newblock \doi{10.1103/PhysRevA.97.012127}.
\newblock URL \url{https://link.aps.org/doi/10.1103/PhysRevA.97.012127}.

\bibitem[Popescu and Rohrlich(1994)]{Popescu:1994aa}
Sandu Popescu and Daniel Rohrlich.
\newblock Quantum nonlocality as an axiom.
\newblock \emph{Foundations of Physics}, 24\penalty0 (3):\penalty0 379--385,
  1994.
\newblock \doi{10.1007/BF02058098}.
\newblock URL \url{https://doi.org/10.1007/BF02058098}.

\bibitem[Riley(2018)]{riley_optics}
Mitchell Riley.
\newblock Categories of optics.
\newblock 2018.
\newblock \doi{10.48550/arXiv.1809.00738}.
\newblock URL \url{https://doi.org/10.48550/arXiv.1809.00738}.

\bibitem[Rom{\'a}n(2020{\natexlab{a}})]{roman_comb}
Mario Rom{\'a}n.
\newblock Comb diagrams for discrete-time feedback, 2020{\natexlab{a}}.
\newblock URL \url{https://doi.org/10.48550/arXiv.2003.06214}.

\bibitem[Rom{\'a}n(2020{\natexlab{b}})]{roman_optics}
Mario Rom{\'a}n.
\newblock Profunctor optics and traversals, 2020{\natexlab{b}}.
\newblock URL \url{https://doi.org/10.48550/arXiv.2001.08045}.

\bibitem[Rom{\'a}n(2021)]{roman_coend}
Mario Rom{\'a}n.
\newblock Open diagrams via coend calculus.
\newblock In \emph{Proceedings ACT 2020}, volume 333, pages 65--78. Open
  Publishing Association, Feb 2021.
\newblock \doi{10.4204/eptcs.333.5}.

\bibitem[Rom{\'a}n(2023)]{roman_thesis}
Mario Rom{\'a}n.
\newblock \emph{Monoidal Context Theory}.
\newblock PhD thesis, Tallinn University of Technology, 2023.
\newblock URL \url{https://doi.org/10.23658/taltech.54/2023}.

\bibitem[Selby et~al.(2021)Selby, Scandolo, and
  Coecke]{Selby2021reconstructing}
John~H. Selby, Carlo~Maria Scandolo, and Bob Coecke.
\newblock Reconstructing quantum theory from diagrammatic postulates.
\newblock \emph{{Quantum}}, 5:\penalty0 445, April 2021.
\newblock ISSN 2521-327X.
\newblock \doi{10.22331/q-2021-04-28-445}.
\newblock URL \url{https://doi.org/10.22331/q-2021-04-28-445}.

\bibitem[Selinger(2008)]{selinger_idempotents}
Peter Selinger.
\newblock Idempotents in dagger categories: (extended abstract).
\newblock \emph{Electronic Notes in Theoretical Computer Science},
  210:\penalty0 107--122, 2008.
\newblock \doi{10.1016/j.entcs.2008.04.021}.

\bibitem[Sengupta(2024)]{sengupta2024achievingmaximalcausalindefiniteness}
Kuntal Sengupta.
\newblock Achieving maximal causal indefiniteness in a maximally nonlocal
  theory, 2024.
\newblock URL \url{https://arxiv.org/abs/2411.04201}.

\bibitem[Simmons and Kissinger(2022)]{SimmonsKissinger2022}
Will Simmons and Aleks Kissinger.
\newblock {Higher-Order Causal Theories Are Models of BV-Logic}.
\newblock In Stefan Szeider, Robert Ganian, and Alexandra Silva, editors,
  \emph{47th International Symposium on Mathematical Foundations of Computer
  Science (MFCS 2022)}, volume 241 of \emph{Leibniz International Proceedings
  in Informatics (LIPIcs)}, pages 80:1--80:14, Dagstuhl, Germany, 2022. Schloss
  Dagstuhl -- Leibniz-Zentrum f{\"u}r Informatik.
\newblock ISBN 978-3-95977-256-3.
\newblock \doi{10.4230/LIPIcs.MFCS.2022.80}.

\bibitem[Simmons and Kissinger(2024)]{simmons_completelogic}
Will Simmons and Aleks Kissinger.
\newblock A complete logic for causal consistency, 2024.

\bibitem[Tambara(2006)]{tambara}
Daisuke Tambara.
\newblock {Distributors on a tensor category}.
\newblock \emph{Hokkaido Mathematical Journal}, 35\penalty0 (2):\penalty0 379
  -- 425, 2006.
\newblock \doi{10.14492/hokmj/1285766362}.

\bibitem[Taranto et~al.(2025)Taranto, Milz, Murao, Quintino, and
  Modi]{taranto2025higherorderquantumoperations}
Philip Taranto, Simon Milz, Mio Murao, Marco~T{\'u}lio Quintino, and Kavan
  Modi.
\newblock Higher-order quantum operations, 2025.
\newblock URL \url{https://arxiv.org/abs/2503.09693}.

\bibitem[van~der Lugt et~al.(2023)van~der Lugt, Barrett, and
  Chiribella]{Lugt:2023aa}
Tein van~der Lugt, Jonathan Barrett, and Giulio Chiribella.
\newblock Device-independent certification of indefinite causal order in the
  quantum switch.
\newblock \emph{Nature Communications}, 14\penalty0 (1):\penalty0 5811, 2023.
\newblock \doi{10.1038/s41467-023-40162-8}.
\newblock URL \url{https://doi.org/10.1038/s41467-023-40162-8}.

\bibitem[Wechs et~al.(2021)Wechs, Dourdent, Abbott, and
  Branciard]{PRXQuantum.2.030335}
Julian Wechs, Hippolyte Dourdent, Alastair~A. Abbott, and Cyril Branciard.
\newblock Quantum circuits with classical versus quantum control of causal
  order.
\newblock \emph{PRX Quantum}, 2:\penalty0 030335, Aug 2021.
\newblock \doi{10.1103/PRXQuantum.2.030335}.
\newblock URL \url{https://link.aps.org/doi/10.1103/PRXQuantum.2.030335}.

\bibitem[Wilde(2013)]{Wilde_2013}
Mark~M. Wilde.
\newblock \emph{Quantum Information Theory}.
\newblock Cambridge University Press, 2013.

\bibitem[Wilson(2023)]{wilson2023compositional}
M.~Wilson.
\newblock \emph{Compositional Frameworks for Supermaps and Causality}.
\newblock Phd thesis, University of Oxford, 2023.

\bibitem[Wilson and Chiribella(2021)]{Wilson_causal}
Matt Wilson and Giulio Chiribella.
\newblock Causality in higher order process theories.
\newblock In \emph{Proceedings QPL 2021}, volume 343, pages 265--300. Open
  Publishing Association, 2021.
\newblock \doi{10.4204/eptcs.343.12}.
\newblock URL \url{http://dx.doi.org/10.4204/EPTCS.343.12}.

\bibitem[Wilson and Chiribella(2022)]{wilson_polycategories}
Matt Wilson and Giulio Chiribella.
\newblock Free polycategories for unitary supermaps of arbitrary dimension,
  2022.
\newblock URL \url{https://doi.org/10.48550/ARXIV.2207.09180}.

\bibitem[Wilson and Ormrod(2023)]{wilson2023originlinearityunitarityquantum}
Matt Wilson and Nick Ormrod.
\newblock On the origin of linearity and unitarity in quantum theory, 2023.
\newblock URL \url{https://arxiv.org/abs/2305.20063}.

\bibitem[Wilson et~al.(2022)Wilson, Chiribella, and Kissinger]{wilson_locality}
Matt Wilson, Giulio Chiribella, and Aleks Kissinger.
\newblock Quantum supermaps are characterized by locality, 2022.
\newblock URL \url{https://doi.org/10.48550/ARXIV.2205.09844}.

\bibitem[Ziman(2008)]{Ziman_2008}
M{\'{a}}rio Ziman.
\newblock Process positive-operator-valued measure: A mathematical framework
  for the description of process tomography experiments.
\newblock \emph{Physical Review A}, 77\penalty0 (6), 6 2008.
\newblock \doi{10.1103/physreva.77.062112}.
\newblock URL \url{https://doi.org/10.1103%2Fphysreva.77.062112}.

\end{thebibliography}
